\documentclass[12pt]{article}
\usepackage{amsmath}
\usepackage{graphicx}
\usepackage{enumerate}
\usepackage{bbm}
\usepackage{natbib}
\usepackage{dsfont}
\usepackage{url} 
\usepackage{relsize,color}
\usepackage[dvipsnames]{xcolor}
\usepackage{setspace}
\usepackage{etoolbox}
\AtBeginEnvironment{tabular}{\onehalfspacing}
\usepackage{geometry}
\usepackage{multirow}
\usepackage{multicol}
\usepackage{bigints}
\usepackage{amsfonts,amsmath,amssymb,amsthm, bbm}
\usepackage{wrapfig}
\usepackage{verbatim,float,url}
\usepackage{graphicx}
\usepackage{sidecap}
\usepackage{subcaption}
\usepackage[skip=0pt]{caption}

\newtheorem{definition}{Definition}

\newtheorem{corollary}{Corollary}
\newtheorem{theorem}{Theorem}
\newtheorem{proposition}{Proposition}

\usepackage{algpseudocode}
\usepackage{algorithm}
\usepackage{epsfig}
\usepackage{ulem} 
\usepackage{xspace} 
\usepackage{graphicx}
\usepackage{sidecap}
\usepackage{times,epsfig,makeidx,amsfonts,amsmath}
\usepackage{lscape} 

\usepackage{collcell}
\makeatletter
\newcolumntype{G}{>{\collectcell\@gobble}c<{\endcollectcell}@{}}
\makeatother

\usepackage{bm}

\usepackage{mathtools}

\usepackage{mdwlist}

\usepackage{tikz}  
\usetikzlibrary{trees}
\usetikzlibrary{arrows.meta}
\usetikzlibrary{decorations.pathreplacing}

\usepackage{mathtools}

\def\E{\mathbb{E}}

\graphicspath{{figs/}}

\newcommand{\iidsim}{\overset{iid}{\sim}} 

\newcommand{\di}{{\rm d}}

\newcommand{\bfe}{\textbf{e}}
\newcommand{\bfY}{\textbf{Y}}
\newcommand{\bfy}{\textbf{y}}

\newcommand{\bfr}{\textbf{r}}

\newcommand{\bftau}{\boldsymbol{\tau}}
\newcommand{\bfgamma}{\boldsymbol{\gamma}}
\newcommand{\bfpi}{\boldsymbol{\pi}}
\newcommand{\bfalpha}{\boldsymbol{\alpha}}

\renewcommand{\&}{and}
\usepackage{ulem}

\usepackage{algorithm}
\usepackage{algpseudocode}

\algnewcommand\algorithmicinput{\textbf{INPUT:}}
\algnewcommand\INPUT{\item[\algorithmicinput]}
\algnewcommand\algorithmicoutput{\textbf{OUTPUT:}}
\algnewcommand\OUTPUT{\item[\algorithmicoutput]}

\DeclareMathOperator*{\argmin}{arg\,min}
\DeclareMathOperator*{\argmax}{arg\,max}

\theoremstyle{plain}
\newtheorem{assumption}{Assumption}
\newtheorem{lemma}{Lemma}
\allowdisplaybreaks
\usepackage{bm}

\usepackage{xr}
\begin{document}

\def\spacingset#1{\renewcommand{\baselinestretch}%
{#1}\small\normalsize} \spacingset{1}

%
%
%
%
%

\title{Bayesian variance change point detection with credible sets}
\author{Lorenzo Cappello, Oscar Hernan Madrid Padilla
	\thanks{L. Cappello is with the Department of Economics and Business, Universitat Pompeu Fabra, Barcelona, Spain, and the Data Science Center at Barcelona School of Economics, Barcelona, Spain. He acknowledges the support of Ayudas Fundacion BBVA a Proyectos de Investigacion Cientifica 2021, the Spanish Ministry of Economy and Competitiveness grant PID2022-138268NB-I00, financed by MCIN/AEI/10.13039/501100011033, FSE+MTM2015-67304-P, and FEDER, EU;  the grant Ramon y Cajal 2022 RYC2022-038467-I, financed by MCIN/AEI/10.13039/501100011033 and FSE+, and the Severo Ochoa Programme for Centres of Excellence in R\&D (Barcelona School of Economics CEX2019-000915-S), funded by MCIN/AEI/10.13039/50110001103}
	\thanks{O.H. Madrid Padilla is with the Department of Statististics \& Data Sciences, UCLA, Los Angeles, USA.  He acknowledges the support of National Science Foundation through the grant DMS-2015489. }
}

\markboth{IEEE TRANSACTIONS ON PATTERN ANALYSIS AND MACHINE INTELLIGENCE}%
{How to Use the IEEEtran \LaTeX \ Templates}

\maketitle

\begin{abstract}
	This paper introduces a novel Bayesian approach to detect changes in the variance of a Gaussian sequence model, focusing on quantifying the uncertainty in the change point locations and providing a scalable algorithm for inference. We do that by framing the problem as a product of multiple single changes in the scale parameter. We fit the model through an iterative procedure similar to what is done for additive models. The novelty is that each iteration returns a probability distribution on time instances, which captures the uncertainty in the change point location. Leveraging a recent result in the literature, we can show that our proposal is a variational approximation of the exact model posterior distribution. We study the convergence of the algorithm and the change point localization rate. Extensive experiments in simulation studies and applications to biological data illustrate the performance of our method. 
\end{abstract}


\section{Introduction}
\label{sec:intro}

The detection of change points -- when and how many times the distribution underlying an ordered data stream experiences a change -- is a field with a long history \cite{page1954continuous,barnard1959control}. The collection of large quantities of data enabled by new technologies -- \textit{e.g.}, wearable devices, telecommunications infrastructure, and genomic data -- has fostered a renaissance of the field. To analyze these data sets, we cannot assume relatively rigid structures where parameters are shared across all observations. Change points define partitions of the data where, within each segment, assumptions like exchangeability are not violated; hence, standard methods can be used.

The vast majority of the available methods output point estimates of the number of change points and the locations on these change points; a nonexhaustive list includes \cite{killick2012optimal,fryzlewicz2014wild,wang2018optimal, baranowski2019narrowest}.	An underdeveloped aspect of change point detection is uncertainty quantification, in the sense of being able to provide a set of times instances containing the location of the change at a prescribed level of significance as done by early works of \cite{worsley1986confidence,siegmund1986boundary}. A popular approach to construct confidence intervals builds on Yao \cite{yao1987approximating}, who derives the asymptotic distribution of the process of tests used to detect a single change in Gaussian mean. The result has then been extended to other settings, such as linear regression \cite{bai1998estimating}, piecewise constant signal plus noise \cite{eichinger2018mosum} and time series \cite{ling2016estimation}. However, the asymptotic distribution is not always known in closed form or may depend on unknown quantities. Recent attempts have addressed this gap taking a multiscale approach -- in the piecewise-constant mean case \cite{fric14,fan20}, and in piecewise-linear mean model \cite{fryzlewicz2020narrowest} -- or a post-selection inference one \cite{jewell2019testing}. \cite{fryzlewicz2020narrowest} is a good reference for the state-of-the-art in uncertainty quantification for change point detection location estimates.

Much of this recent literature has not covered the case where we are not interested in detecting a change in the mean of a sequence but changes in the underlying variance. There are methods returning point estimates for this task, 
such as the cumulative sum squares \cite{inclan1994use}, penalized weighted least squares methods \cite{chen1997testing,gao19}, the fused lasso \cite{padilla2022variance}, and PELT \cite{killick2012optimal}. However, none returns confidence or credible sets along with the point estimates. Here we introduce a simple and computationally scalable approach that helps address the issue.


The need to add a measure of uncertainty associated with changes in variance is motivated by experimental data presented by \cite{gao19}, who study a new technique to determine whether a liver is viable for transplant or not. Methods that are routinely employed involve a high degree of subjectivity (\textit{e.g.}, visual inspection by medical personnel) or invasive techniques (\textit{e.g.}, biopsy) with the risk of damaging the organ. The new procedure consists in monitoring surface temperature fluctuations of the organ at multiple locations using a temperature-controlled perfusion liquid. High-temperature fluctuations suggest a responsive, hence viable liver, and low variations indicate the loss of viability.  
A feature of the data that stands out is that the mean change smoothly. In \cite{gao19}, the authors propose a point estimator of a single change in variance that accounts for such smooth mean trend. We argue that, in such sensitive applications, a measure of uncertainty is as important as the ability to detect the change point. There are many other sensitive applications where such a feature is desirable, such as neuroscience and seismology.


Bayesian change point methods offer a natural way to quantify uncertainty \cite{chernoff1964estimating,smi75,bar92,bar93,carlin1992hierarchical} and have been linked to robustness to model misspecification \cite{lai2011simple,cap21}. Despite this obvious benefit, the Bayesian literature has not kept pace with recent advances in the literature, and practitioners do not commonly employ these methods. The main reasons are the high computational burden required and the limited literature on statistical guarantees available for these methods. In the applications considered, such limitations are critical. A lack of theoretical guarantees is not desirable in high-stakes settings. Furthermore, while the sample size at a point in the organ is relatively small, one must repeat the analysis at thousands of locations, one for each point where the temperature is monitored. Similar issues arise in other settings as well.

The high run times of Bayesian change point methods are primarily due to the Markov chain Monte Carlo (MCMC). There are a few computational speed-up, including closed-form recursions that exploit conjugate priors  \cite{fea06,lai2011simple}, Empirical Bayes approaches \cite{liu20}, and approximate recursions \cite{cap21}. However, despite the improvements, Bayesian change point methods remain orders of magnitude slower than state-of-the-art approaches, even for small sample sizes. 

A statistical property that researchers seek in a change point method is the localization rate. The literature on this topic for Bayesian methods is minimal, with few works dealing with optimality in a minimax sense \cite{liu20,cap21}. A second way of looking at statistical guarantees is the trustworthiness of the algorithms used for inference. As sample size and number of change points grow, one questions the feasibility of MCMC chains to explore the state space fully, especially when the number of change points is unknown. In simulations, Cappello et al. \cite{cap21} illustrate a standard Gaussian piecewise-constant mean scenario (BLOCKS, \cite{fryzlewicz2014wild}), where MCMC chains fail to converge lacking a ``good" initialization. Similar concerns were raised in the high-dimensional variable selection literature \cite{chen2019fast} and crossed random effects models \cite{gao2017efficient}.



In a way, our proposal targets both issues as we provide a Bayesian variance change point detection method that comes with theoretical guarantees, and gives inference in linear time without requiring MCMC. These features are essential for the applicability of our method in modern applications requiring detection of change points with uncertainty quantification.

\subsection{Our contributions and related work}

We start by considering a simple setting with an ordered sequence of $T$ independent Gaussian random variables with constant mean undergoing $K$ changes in variance. In the presence of a single change point ($K=1$), one can construct a Bayesian model with conjugate priors using a latent random variable indexing the unknown location of the change point \cite{smi75,raftery1986bayesian}. Such a model inherently describes uncertainty on the change point locations through the posterior distribution of the latent variable. The computational cost is minimal, being the update of posterior parameters.

The generalization to multiple change points inflates the computational costs because closed-form updates are unavailable, and the posterior distribution needs to be approximated. The computational advantages described for the single change point (or single effect) model are essentially lost, except for the possibility of writing Gibbs sampler full conditionals. 
Ideas from additive models and variable selection suggest that it is possible to preserve the advantages of single-effect models if one ``stacks" multiple single-effect models and solves them recursively, one at a time. \cite{has00} provide the first link between additive models and posterior approximation. Recently, \cite{wang2020simple} employ such an approach in variable selection for the linear model and suggest it can be used to detect changes of a piecewise constant Gaussian mean. They also show that this recursive algorithm is essentially a variational approximation to the actual posterior distribution.

Our work builds on these ideas. However, an additive structure is unsuitable for variance parameter changes; thus, a different construction is necessary. The building block of our proposal is a single change point model with a random change point location and a random scale parameter that multiplies a baseline variance. \textit{I.e.}, we have a nested structure where the variance to the left of the change point is ``scaled" to define the variance to the right. Such a construction is essential to generalize the model to multiple change points. In the paper, we study the theoretical properties of the single change point model and show that it attains consistency for the task of  localization of the changes in variance \cite{wan21}. The proof holds both in the $i.i.d.$ setting, when there is serial dependence, and does not rely on the Gaussian assumption. To our knowledge, it is the first proof of a Bayesian estimator finite sample guarantees for changes in variance.

To extend the single effect model to the multiple change point scenarios, we consider multiple independent replicates of the model described above, \textit{i.e.}, multiple latent indicators describing change point locations and multiple scale parameters. Rather than summing multiple single effects as in \cite{wang2020simple}, we take a product of the single effects; in practice, it corresponds to taking a product of scale parameters. Such construction is new and it is essential to develop a scalable algorithm. We then propose a recursive algorithm to fit this product of models similar to those used for additive models. The algorithm has a linear computational time in the number of change points. Building on the intuition in \cite{wang2020simple}, we can prove that our algorithm corresponds to a Variational Bayes (VB) approximation to the actual model posterior. Such characterization has the advantage of linking our proposal to existing algorithms for posterior approximation and allows us to prove the algorithm's convergence to a stationary point. However, the characterization is just a nice byproduct of our novel modular construction. If one were to do a VB approximation of existing Bayesian change point detection methods (\textit{e.g.}, \cite{bar92}), one would get a different algorithm.

The approximate posterior distribution naturally allows us to obtain point estimates for the location of the change points and to construct credible sets describing the uncertainty underlying these estimates at a prescribed level. These sets are a discrete set of time instances chosen through their posterior probability and are not necessarily intervals, as opposed to recent works \cite{fric14,fryzlewicz2020narrowest}. VB posterior approximations are known to provide excellent point estimates but to underestimate uncertainty \cite{bishop2006pattern}. In simulations, we show that our method offers point estimates as accurate as state-of-the-art methodologies (at times even more accurate), and the uncertainty underestimation typical of VB is not extremely severe. Overall, at little additional computational costs, our proposal provides points estimates as precise as those of competitors and a measure of uncertainty that competitors lack.

Our proposal is modular and can be easily generalized to more complex and realistic data-generating mechanisms. As proof, we show how to extend our base method to variance change point detection in the presence of autoregression, a smoothly varying mean trend (\textit{e.g.},  \cite{gao19}), and non-Gaussian setting. We also show that it can accommodate a situation where multiple observations are available per time point, which is relevant when data are binned (\textit{e.g.}, \cite{cap21}) or when we have cyclical data (\textit{e.g.}, \cite{ush22}). To the best of our knowledge, our proposals are the first to tackle some of these tasks. These extensions are a proof-of-concept of our approach's generalizability. The paper includes a description and empirical study of each, and future work will study their properties.


To summarize, the paper includes the following contributions. (a) We study the theoretical properties of a Bayesian variance single change point estimator and establish that it is consistent in settings with dependent and non-Gaussian data. (b) We propose a new methodology for Bayesian variance change point detection when there are multiple change points. The method is fast, empirically accurate, and allows the construction of credible sets describing the uncertainty of change point location. (c) We justify the algorithm used to fit our methodology and establish its convergence. (d) We show that we can generalize our proposals to realistic settings, such as autoregression, smoothly varying mean, misspecified model and repeated measurements. (e) We study the new experimental technique to assess liver procurement and provide a measure of uncertainty.

The rest of the paper proceeds as follows. Section~\ref{sec:single} introduces the single change point model and study its theoretical properties. In Section~\ref{sec:prisca}, we extend the single change point model to multiple changes, introduce an algorithm to approximate the posterior distribution and establish the theoretical underpinnings of the algorithm. After having introduced the baseline model, Section~\ref{sec:ext} outlines several extensions. Section~\ref{sec:sim} details a simulation study comparing our proposal's performance vis-a-vis alternatives. In Section~\ref{sec:liver}, we analyze the liver procurement data of \cite{gao19}. In Section~\ref{sec:oce}, we include a second application to new oceanographic data. Section~\ref{sec:disc} concludes. An implementation of the proposed method and the code to reproduce the numerical experiments are available as a  \texttt{R} package for download at \url{https://github.com/lorenzocapp/prisca}.

\section{Model for a single change in variance}\label{sec:single}

Assume we observe a vector of independent Gaussian random variables such that there is a time instance $t_0$ that partitions the vector into two segments: for $t<t_0$, $Y_t |\sigma_l^2 \iidsim N(0,\sigma_l^2)$,  for $t \geq t_0$, $Y_{t} |\sigma_r^2 \iidsim N(0,\sigma^2_r)$. We are primarily interested in an estimate of the unknown location $t_0$ and a measure of the uncertainty of such an estimate. The parameters $\sigma_l^2$ and $\sigma_r^2$ could be known or unknown. We consider without loss of generality (w.l.o.g.) a setting where $\sigma_l^2$ is known, and $\sigma_r^2$ is unknown. The extension to multiple change points will circumvent this assumption. A natural model for this question is the following: 

\begin{align} \label{eq:single}
	\bfY | \bftau, \sigma^2 ,\bfgamma &\sim   \bftau^{-1} \circ  \bfe \,\,\,\,\, \text{with} \,\, \bfe \sim N_T(0,\sigma^2 I_T), \nonumber \\
	\bfgamma|\bfpi &\sim  \text{Multinomial(1,}\bfpi),\\
	\tau_t| \bfgamma  \,&=\,\begin{cases}
		1 &  \text{if} \,\,\,\, 1\leq t < t_1 ,\\
		s & \text{if} \,\,\,\, t_1\leq t \leq T, 
	\end{cases} \nonumber\\
	s^2 | a_0 &\sim \text{Gamma}(a_0, a_0), \nonumber
\end{align}
where, $\bfgamma$ denotes a Multinomial random vector of length $T$, with an event corresponding to each time instance, and $t_1$ denotes the time instance sampled, $\tau_t$ an entry of the $T$-length vector $\bftau$, $\circ$ the Hadamard product, and with abuse of notation $\mathbf{x}^{k}=(x_t^{k})_{1:T}$ denotes the elementwise power of vector entries.

The random vector $\bfgamma$ describes the unknown location of the change point. The parameter $\bfpi$ of the Multinomial distribution describes the prior probability of having a change in variance at any given instance $t$. The Multinomial vector has a single trial since we are considering the single change point case. The model represents $(\sigma_l^2,\sigma_r^2)$ through a baseline variance $\sigma^2$ (assumed known) which gets scaled by  $s$, such that to the right of $t_1$, the $Y_i$'s are Gaussian distributed with variance $s^{-2}\sigma^2$. To the left of $t_1$, we have a neutral model with $\tau_t=1$ and variance equal to $\sigma^2$. 
The Gamma distribution has equal shape and rate parameters to have a priori expectation equal to one, which roughly corresponds to a null model, and it is convenient to have one less parameter to tune.


The distributions in \eqref{eq:single} are conjugate, so that posterior distribution $P(\bfgamma, \bftau| \bfy)$ (for parsimony we omit the hyperparameters) is available in closed-form:
\begin{align}\label{eq:posterior}
	\bfgamma|\bfy,\bfpi, \sigma^2,a_0 &\sim  \text{Multinomial}(1,\bfalpha),\\
	s^2 | \bfy,\bfpi, \sigma^2,\gamma_t=1 &\sim \text{Gamma}(a_t,b_t), \nonumber
\end{align}
where $\bfalpha$ is a vector with entries
\begin{equation}\label{eq:postalpha}
	\alpha_t:=P(\gamma_t=1 | \bfy,\bfpi, \sigma^2,a_0)= \frac{P(\bfy | \gamma_t=1,\bfpi, \sigma^2,a_0) \pi_t}{\sum_{\pi_j} P(\bfy | \gamma_j=1,\bfpi, \sigma^2,a_0) \pi_j},
\end{equation}
and $P(\bfy | \gamma_t=1,\bfpi, \sigma^2,a_0)=\int P(\bfy | \gamma_t=1,\bfpi, \sigma^2,s) \di P(s|a_0)$. The posterior hyperparameters of the Gamma distribution are for $t=1,\ldots,T$
\begin{equation}\label{eq:posttau}
	a_t =a_0 + \frac{T-t+1}{2} \,\,\,\, \text{and} \,\,\,\, b_t=a_0+ \frac{\bfy_{t:T}^T \bfy_{t:T}}{2 \sigma^2}.
\end{equation}
All quantities in \eqref{eq:posterior}-\eqref{eq:posttau} are available in closed-form (no numerical approximation required; see Supplementary Material S1 for the derivations). If one works with a prior on $s$ different than Gamma, closed-form posterior quantities are not available and numerical approximations are necessary; see Section~\ref{sec:ext} for details.

What makes model \eqref{eq:single} appealing is the vector of probabilities $\bfalpha$ can be used simultaneously for point estimation and uncertainty quantification. This differs from recent proposals where the two tasks are disentangled. 
An obvious point estimate is the maximum a posterior
\begin{equation}\label{eq:point}
	\widehat{t}= \argmax_{t \in T} \alpha_t.
\end{equation}
The next subsection studies the properties of such point estimates. In addition, we can return a set of time instances containing the true change point with a prescribed probability level, \textit{i.e.} a credible set\textit{}. 
The obvious way to construct the credible set is to rank $\alpha_t$ in decreasing order and choose the smallest number of time instances such that the sum of the posterior probabilities is bigger than a prescribed probability $p$. Thus, the resulting set is
\begin{equation}\label{eq:cs}
	\mathcal{CS}(\bfalpha, p):=\arg \min_{S \subset \{1,\ldots,T\}:\sum_{t\in S} \alpha_t>p}  | S |.
\end{equation}
We highlight that the credible set built through the above is not necessarily an interval. Figure~\ref{fig:example1} column (A) depicts an example of what was discussed.


\begin{figure}
	\centering
	\hspace{-0.2cm}\includegraphics[scale=0.7]{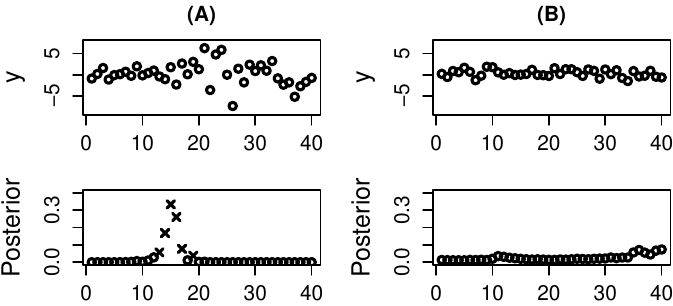}
	\caption{\small{\textbf{Example of single (A) and no (B) change point.} (A) top left panel depicts a sample of $T=40$ Gaussian random variables, with $t_0=16$, $\sigma_l=1$ and $\sigma^2_r=3$. The bottom panel depicts the posterior distribution of $\bfgamma$, \textit{i.e.}, $\bfalpha$. Cross dots depict $\mathcal{CS}(\bfalpha,.9)$, $\widehat{t}=15$. (B) depicts an example with no change point: top panel $T=40$ samples from a zero mean, variance one Gaussian; bottom panel as in (A)}}
	\label{fig:example1}
\end{figure}

We also have a closed expression for the posterior expectation of $\bftau$, which provides an estimates of the variance at all time instances. Given $\bfalpha$, we can write $\overline{\bftau}^2:=E[\bftau^2|\bfy]$, where
\begin{align}\label{eq:expectau}
	\overline{\bftau}^2=&(\alpha_1 \hat{s}_1^2 + 1-\alpha_1,  \alpha_1 \hat{s}_1^2+ \alpha_2 \hat{s}_2^2 + 1-\alpha_1-\alpha_2,\ldots \\
	& \ldots,\sum_{i=1}^{t} \alpha_i \hat{s}_i^2 + 1 -\sum_{i=1}^{t} \alpha_i,\ldots, \sum_{i=1}^T \alpha_i \hat{s}_i^2), \nonumber
\end{align}
with $\hat{s}_t^2:=\E[s|\bfy,\gamma_t=1]=a_t/b_t$, the posterior expected precision when fitting a model conditionally on $\gamma_t=1$. The $t$th entry of $\overline{\bftau}^2$ is a weighted average of all the models that can be assigned at time instance $t$. For example, the first time instance has the ``neutral model" ($\tau_1$) always assigned except when $\gamma_1=1$; in $t_2$, $\tau_2$ is different from one only when $\gamma_1=1$ or $\gamma_2=1$; in $t_T$, the posterior expectation is the weighted average of the $T$ posteriors expectations $(\hat{s}_t^2)_{1:T}$. We refer to the Supplementary Material for closed-form calculations.

The model described here is closely related to a general approach for Bayesian detection of a single change point described by Smith \cite{smi75}. A difference is that Smith models the variance to the left and right of $t_0$ as separate random variables. Here, we have a nested structure, with $s$ tilting a baseline (known) variance $\sigma^2$ to describe $\sigma_r^2$. This feature is useful in the generalization to multiple change points. 

Throughout the section, the model assumes the existence of a change in the variance. This is standard in change point detection. However, it is natural to wonder what the posterior of $\bfgamma$ looks like if there is no such change. Our formulation includes a realistic description of the null model at any time instance: lacking evidence of a change, the posterior $\tau_t$ will concentrate at one for all $1\leq t \leq T$. A consequence of such construction is that we empirically observe $\bfalpha$ is somewhat diffuse when there is no change point; see, for example, Figure~\ref{fig:example1} column (B). This is not the case in the traditional Bayesian setup, see \cite{smi75}. Such observation will be crucial in discussing the extension to multiple change points.

\subsection{Theory}
\label{sec:theory}

We consider the change point selection criterion defined in \eqref{eq:point}. We study the localization rate of this point estimator. The main result is based on the following assumption which requires the definition of $\alpha$-mixing given in Definition S1 in the Supplementary Material.

\begin{assumption} \label{ass1}
	Let $t_0$ be the time instance such that $Y_{t} \sim  F_1$ for $t \geq t_0$ and $Y_{t} \sim F_0$ for $t < t_0$, and let $\tau^2=\sigma^2_l/\sigma^2_r$ where  $\sigma^2_l:= \mathbb{E}(Y_1^2)$ and $\sigma^2_r := \mathbb{E}(Y_{t_0+1}^2) $.  We assume that $\{Y_t, t \in  \mathbb{N}\}$ is $\alpha $-mixing with coefficients satisfying $\alpha_k \,\leq \, e^{-C k }$ for some positive constant $C>0$.  Also, we require that   $\max\{ \mathbb{E}(| Y_1^2- \sigma_l^2 |^{2+\delta  }),  \mathbb{E}(| Y_{t_0+1}^2-  \sigma_r^2|^{2+\delta  }) \} \,\leq \,D_1$ for some positive constants $\delta,D_1$. In addition assume that: 
	\begin{enumerate}[a.]
		\item There exists a constant $c>0$ such that $\min\{t_0,T-t_0\}>cT$.
		\item For some  fixed intervals $I_1 \subset (1,\infty)$  and $I_2 \subset (0,1)$   we have  that  $\tau^2 \in I_1\cup I_2$.  
		\item The hyperparameters are $a_0>0$ and $\bfpi$ satisfies that $\sum_t \pi_t=1$, and  $ c_{\min}/T \leq \pi_t \leq c_{\max}/T$ for all $t =1,\ldots,T$, and some positive constants $c_{\min}$ and $c_{\max}$.
	\end{enumerate}
\end{assumption}


Assumption \ref{ass1} allows the data to be dependent by relaxing the usual i.i.d. condition with an $\alpha$-mixing condition. Moreover, we do  not require the data to be Gaussian, instead we only require a moment condition on the measurements. Furthermore,  Assumption~\ref{ass1} \textit{a.} has appeared in the literature, see for instance Theorem 2 in \cite{cap21}. Assumption~\ref{ass1} \textit{b.} requires the true scaling of the variance to be non-negligible: $1$ is excluded from $I_1$ and $I_2$ but it can be close to them. Assumption~\ref{ass1} \textit{c.} requires that the priors are proper. We are ready to state the main result of the section.

\begin{theorem}\label{thm:single}
	Suppose that Assumption~\ref{ass1} holds. Then, for $\epsilon>0$ there exists a constant $c_1>0$ such that, with probability approaching one, we have that
	\begin{equation}
		\label{eqn:conclusion}
		\max_{t:\min\{t,T-t\}>cT, |t-t_0|>c_1 \sqrt{T }\log^{1+\epsilon} T }\alpha_t < \alpha_{t_0}.
	\end{equation}
	If in addition, we assume that $Y_{t} \iidsim N(0,\sigma_r^2)$ for $t \geq t_0$ and $Y_{t} \iidsim N(0,\sigma_l^2)$ for $t < t_0$, then with probability approaching one, we have that
	$$ \max_{t:\min\{t,T-t\}>cT, |t-t_0|>c_1 \sqrt{T \log^{1+\epsilon} T }}\alpha_t < \alpha_{t_0}.$$
\end{theorem}
Notably, the first part of Theorem~\ref{thm:single} establishes that the estimator~\eqref{eq:point}, constructed based on the model~\eqref{eq:single}, achieves a localization rate of order $\sqrt{T} \log^{1+\epsilon}T$, even in scenarios where the data can exhibit dependencies and non-Gaussian distributions. The second part of Theorem~\ref{thm:single} then states that if the data is independent and Gaussian, the rate becomes $\sqrt{T \log^{1+\epsilon}T}$. It is worth highlighting that this rate matches the one achieved by the Binary Segmentation method introduced by \cite{wan21}. Consequently, in the single change point setting, our algorithm showcases consistency in more general settings compared to Binary Segmentation, and it achieves an identical rate as Binary Segmentation in the case of independent Gaussian data. However, in the scenario of independent Gaussian data, it remains unknown whether our estimator is minimax optimal as the Wild Binary Segmentation algorithm from \cite{wan21}.




Moreover, we highlight that our algorithm's credible sets are valid in the Bayesian sense. 
Also, an immediate consequence of Theorem \ref{thm:single} is that, with high probability, a point within a distance $\sqrt{T }\log^{1+\epsilon}T$ of the true change point will be contained in the credible set. 

Finally, we conclude with an extension of Theorem \ref{thm:single} to the case of smoothly varying mean.

\begin{corollary}
	\label{cor2}
	Suppose that $\{Y_t\}$  satisfies Assumption \ref{ass1} with $\mathbb{E}(Y_t) = 0$  for all $t$. However, assume that instead of observing  $\{Y_t\}$ we observe $\{Z_t\}$ where $Z_t = \mu_t + Y_t$ for an unknown deterministic sequence $\{\mu_t\}$.   Suppose that $\{\hat{\mu}_t\}$ is an estimator of $\{ \mu_t\}$  independent of  $\{ Y_t\}$. In addition suppose that for a constant $c>0$ it holds that 
	\[
	\sum_{t=1}^T ( \mu_t -\hat{\mu}_t )^2  \,\leq\, c\sqrt{T} \log T , \,\,\,\,\, \underset{t=1,\ldots,T}{\max} \, \vert \hat{\mu}_t\vert \,\leq\,  c 
	\]
	with high probability.  Then, using $\{\hat{Y}_t\}$, with $\hat{Y}_t = Z_t - \hat{\mu}_t $,  as the input data for our method implies that (\ref{eqn:conclusion}) holds with high probability. 
\end{corollary}

Thus,  Corollary \ref{cor2}  allows for changing mean in the model with the main requirement that the changing mean can be estimated at a faster rate, in terms of mean squared error,  than $ T^{-1/2}\log T$, which is the case for nonparametric  estimators such as trend filtering \cite{tib14,madrid2024change}.

Regarding the independence assumption between $\{\hat{\mu}_t\}$ and $\{ Y_t\}$, consider first the scenario where the measurements ${Y_t}$ are independent. In this case, we can partition the data into two subsets based on odd and even time points $t$. One subset can be renamed $\{Y_t\}$, while the other can be used to estimate the mean $\{\hat{\mu}_t\}$, using any nonparametric method. This ensures that $\{\hat{\mu}_t\}$ is independent of $\{Y_t\}$, with $\hat{\mu}_t$ serving as an estimate of $\mu_t$. Here, the index $t$ can be treated as a fixed covariate corresponding to $t/T$.

In the case where $\{Y_t\}$ are not independent, the following approach can be considered. If a second independent dataset is available, it can be used to estimate $\{\hat{\mu}_t\}$ independently of ${Y_t}$. If no second dataset is available, we can split the original data into two subsets: $\{Y_{2ih}\}$ and $\{ Y_{(2i+1)h}\}$, where $h$ is a positive integer. When the original data exhibit weak dependence and $h$ is sufficiently large, this partitioning results in two approximately independent datasets.

Finally, since the estimator proposed in this paper is novel, the analysis done in Theorem 1 and Corollary 1 is not based on existing works. In particular, the proof in the Supplementary Material works explicitly with the statistics associated with our method, performing  multiple careful Taylor's series expansions, and using concentration results for $\alpha$-mixing sequences, see Lemma 1 in the Supplementary Material.

\section{Product of single change point models}\label{sec:prisca}

\subsection{Multiple change points in variance}

Let $\mathcal{C}=\{t_1^*,\ldots,t_K^*\}$ a set of $K$ times instances that partitions a sequence of random variables in $K+1$ segments and let $\sigma^2_1,\ldots,\sigma_{K+1}^2$ denote the variance within each segment. Assume we observe a Gaussian sequence of random variables such that $Y_t \iidsim N(0,\sigma_{i+1}^2)$ for all $t^*_i \leq t < t^*_{i+1}$ and $0\leq i \leq K$, where $t_0^*=1$ and $t_{K+1}^* = T$.  We introduce an approach to detect multiple changes in variance and construct credible sets paired with each change point detected.  

Model \eqref{eq:single} performs this task when $K=1$. An obvious generalization is to have $\bfgamma$ sample more than one point, leading to a vector $\bfgamma$ with multiple distinct non-zero components. This can be done with a multinomial with the same parameters $\bfpi$ and the number of experiments equal to the number of change points. Being the latter unknown, one could either place a prior distribution on it or do model selection using some information criteria. Such a model induces a simple distribution on partitions. Refinements of this rationale, such as assigning higher prior probability to partitions having a given minimum spacing between change points have been developed; \textit{e.g.}, product partition models (PPM) \cite{bar92,bar93}. 

Despite the existence of efficient sampling-based algorithms for PPMs \cite{fea06}, the tractability of the single-change point model is lost and approximating the posterior is costly. 
While it is generally the case that a more complex Bayesian method leads to a higher computational burden, in the linear model literature, \cite{wang2020simple} proposed a model that largely preserves the tractability of single-effect models despite the introduction of multiple effects. The idea is reminiscent of additive models, where multiple single effects models are ``summed together". 


Our proposal borrows this intuition and employs the single change point model \eqref{eq:single} as the building block for the multiple change point extension. We do not have an additive structure because we deal with a scale parameter. Instead of summing, our idea is to  ``stack"  multiple single-effect models and multiply them. We call our proposal PRoduct of Individual SCAle parameters (PRISCA). More formally, let $L$ be an upper bound to the number of possible change points (more discussion  on choosing $L$ later); our Bayesian model can be written as

\begin{align} \label{eq:prisca}
	\bfY | \bftau, \sigma^2 ,\bfgamma, \bftau &\sim  \bftau^{-1} \circ  \bfe \,\,\,\,\, \text{with} \\
	\bfe \sim N_T(0,\sigma^2 I_T) \,\,  &\text{and}\,\,  \bftau= \prod_{l=1}^L \bftau_{l}^{-1}, \nonumber \\
	\bfgamma_l &\sim  \text{Multinomial}(1, \bfpi),\\
	\tau_{t,l}| \bfgamma_l  \,&:=\,\begin{cases}
		1 &  \text{if} \,\,\,\, 1\leq t < t_l ,\\
		s_l & \text{if} \,\,\,\, t_l\leq t \leq T. 
	\end{cases} \nonumber \\
	s_l^2 | a_0 &\sim \text{Gamma}(a_0, a_0), \nonumber
\end{align}
for $ l=1,\ldots,L$, where $\prod_{l=1}^L \bftau_{l}^{-1}$ stands for the element-wise product  of the $L$ vectors,  $t_l$ the index where $\bfgamma_l$ takes value one, 	$\tau_{t,l}$ is one of the entries of the vector $\bftau_l$. As in the single change point case, we describe the conjugate case here, but everything will apply to non-conjugate priors on $s$; see Section~\ref{sec:ext}.

Model \eqref{eq:prisca} represents $\sigma^2_1,\ldots,\sigma_{K+1}^2$ via a baseline variance $\sigma^2$ that gets progressively scaled by $s_1,\ldots,s_L$ (any ordering is possible). For example, let $(\sigma_1^2,\sigma_2^2,\sigma_3^2)$ be the variance of three consecutive segments, PRISCA describes this as $(\sigma^2,s_1^{-2}\sigma^2,s_1^{-2} s_2^{-2}\sigma^2)$ with $s_1^2=\sigma^2_1/\sigma^2_2$ and $s_2^2=\sigma^2_2/\sigma^2_3$. While this may appear convoluted, the benefit of such an approach will be evident when describing the algorithm we use to fit the model.  The assumption that the baseline variance is known (\textit{e.g.}, equal to one) is not a limitation because the method estimates a change point at the start of the sequence.

For $L=1$, the model reduces to the single change point model. This suggests that much of what was discussed in the previous section holds for PRISCA. Each pair $(\bfgamma_l,\bftau_l)$ can model only up to one change point, meaning that, given the posterior distribution of   $\bfgamma_l$, one can construct a point estimate and a credible set through \eqref{eq:point}-\eqref{eq:cs}. Similarly, the posterior expectation $\overline{\bftau_l}^2:=E[\bftau_l^2|\bfy]$ is given by \eqref{eq:expectau}. Since each $\bfgamma_l$ models a single change, $L$ must be larger than or equal to $K$ to be able to detect all the change points. Given the posterior distribution $p(\bftau^2_{1:L},\bfgamma_{1:L}|\bfy,\sigma^2)$, our approach can output multiple point estimates and credible sets, each constructed using the marginal $p(\bftau^2_{l},\bfgamma_{l}|\bfy,\sigma^2)$ for $l=1,\ldots,L$. Setting $L>K$ is not an issue, since we observe empirically that the posterior parameters of the redundant vectors $\bfalpha_l$ are somewhat diffuse and do not concentrate posterior mass on any particular time instance, suggesting that no additional change point is detected; see simulation study and Figure~\ref{fig:example1} column (B).  

Model~\eqref{eq:prisca} is fully symmetric in the $L$ components, a feature that is crucial for the development of the algorithm used for inference. Also note that, as opposed to PPM, model~\eqref{eq:prisca} does not enforce any minimum spacing between change points,  and no mechanism prevents the same time instance to be selected twice; \textit{i.e.}, $p(\bftau^2_{l},\bfgamma_{l}|\bfy,\sigma^2)$ and $p(\bftau^2_{l'},\bfgamma_{l'}|\bfy,\sigma^2)$ concentrates around the same time instance for $l\neq l'$. Again, such choice is motivated by the need to develop a fast algorithm for inference. We will see that such choices do not undermine the performance of the algorithms while they allow our algorithm to be magnitude of orders faster than those used for PPM. We will return to this issue in the next subsection.


\begin{figure}
	\centering
	\includegraphics[width=.45\textwidth]{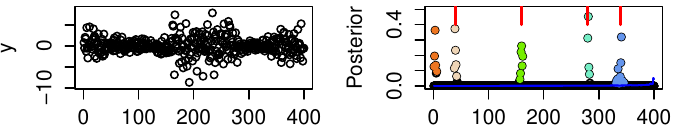}
	\caption{\small{\textbf{Example multiple change points.} We simulate $T=400$ independent Gaussian random variables with four changes in variance ($K=4$, red segments depict true change points locations) and mean zero (left panel). We fit PRISCA with $L=8$ and $a_0=.001$ using Algorithm~\ref{alg:prisca_fit}.  Figure right panel depicts the vectors $\bfalpha_l$ for $l=1$ to $8$. The colored dots depict the five $\mathcal{CS}(.9)$ constructed. The remaining three $\bfalpha$ are too diffuse (depicted with the solid blue lines)}. }
	\label{fig:example_multi}
\end{figure}

The posterior distribution of PRISCA is not available in a tractable form, and we will approximate it with the algorithm discussed next. Now, we illustrate the discussion above with an example. We simulate a sequence of Gaussian random variables with four changes in variance, mean zero, and $\sigma^2=2$ (details in Figure~\ref{fig:example_multi} caption), and approximate the posterior distribution setting $L=8$. Figure~\ref{fig:example_multi} depicts the eight vectors $\bfalpha_{1:8}$ (one of top each other); the approximate posterior distributions $[p(\bfgamma_{l}|\bfy)]_{1:8}$. We colored five separate $90\%$ credible sets that have been constructed for the vector $\bfalpha_l$ that ``concentrate" around a time instance (more details in the next subsection). Note that while $K=4$, PRISCA detected an extra time change at the beginning of the sequence. The reason is that our implementation assumes that the baseline variance $\sigma^2$ is equal to one. Since this is not the case in this example, the algorithm correctly inferred a change at the beginning of the sequence. Three vectors were redundant because there was no remaining effect to capture. We observe here that they ``do not concentrate" around any instance, \textit{i.e.}, they are diffuse (full posterior depicted with the solid blue lines).


In this section, we presented the most parsimonious version of PRISCA, where the same hyperparameters are shared across components. A more general description of the model entails hyperparameters specific to each component, \textit{i.e.}, $(a_{0,l})_{1:L}$ and $(\pi_{l})_{1:L}$. The benefit of the current description is that parameter tuning is minimal, with $L$ the only truly relevant parameter; the next subsection describes a heuristic to choose it. In Section~\ref{sec:sim_robu}, we study the robustness of the algorithm with respect to the choice of $a_0$.

\subsection{How to fit PRISCA}
\label{sec:prisca_algo}

A Gibbs sampler for PRISCA is readily available because the prior distributions in \eqref{eq:prisca} are conditionally conjugate. However, such a chain is poised to mix very slowly because the random vectors $\bfgamma_l$ are highly dependent. The deployment of MCMC in Bayesian change point detection commonly suffers from this problem, and lacking a very good initialization (\textit{e.g.}, the output of another change point detection procedure), chains often fail to converge; see \cite{cap21}. 

The conditionally conjugate structure that makes writing a Gibbs sampler possible suggests a link to techniques used to solve additive models, as first noted by \cite{has00} and more recently by \cite{wang2020simple}. By construction, if one somehow removes the randomness of $(\gamma_{l'},\tau_{l'})_{l'\neq l}$, the posterior update of $(\gamma_l,\tau_l)$ is the same as in the single change point problem. In a Gibbs sampler, the update relies on conditioning on the latest parameters' update; in an additive model, the update relies on being able to compute residuals using previous iterations. In our setting, if we had access to $(\overline{\tau}^2_{l'})_{l'\neq l}$, the squared residuals of model  \eqref{eq:prisca} would be $\overline{r}_l^2= \bfy^2 \circ \prod_{l'\neq l} \overline{\tau}_l^2$, and one would obtain the posterior distribution $\mathrm{pr}(\gamma_l,\tau_l |\overline{r}_l)$ simply through \eqref{eq:postalpha}-\eqref{eq:posttau}. Clearly, $p(\gamma_l,\tau_l |\overline{r}_l)$ is not $p(\gamma_l,\tau_l |\bfy)$ if  $L$  is bigger than one but it is at least a reasonable approximation, and it saves us a lot of computations because it bypasses  the intractability of $p(\gamma_l,\tau_l | \bfy)$. We will see that it is also a good approximation. The underlying reasoning is that for $L=1$ the update is exact; \textit{i.e.} no approximation is necessary (in Section~\ref{sec:sim_robu}, we study this statement numerically). Quite essentially, $p(\gamma_l,\tau_l |\overline{r}_l)$ allows for approximating $\overline{\tau}^2_{l}$ through \eqref{eq:expectau}, which in turns can be used to compute residuals for other effects. This reasoning suggests an iterative algorithm to approximate the posterior of PRISCA.

Algorithm~\ref{alg:prisca_fit} describes the recursion. It requires an initialization for $(\overline{\bftau}^2_{l})_{1:L}$ because one does not have access to them. For simplicity, we set them equal to vectors of ones, \textit{i.e.}, null effects. Then iteratively, the posterior distribution of each effect $(\bftau_l, \bfgamma_l)$ is computed using the single change point updates \eqref{eq:postalpha}-\eqref{eq:posttau}, with the difference that one uses  residuals $\overline{\bfr}_l$ in lieu of $\bfy$. Since we are not considering an additive model, we will not obtain the residuals by subtracting the expectation of each effect (except the one we are updating); instead, we scale $\bfy$ using $(p(\bfgamma_i,\bftau_i |\overline{\bfr}_i))_{i\neq l}$. The rest proceeds in the same fashion as additive models, employing backfitting \cite{fri81,bre85} to reduce the dependency on the order we fit each effect. In this respect, the symmetry in the model construction is essential as we revisit the estimates made by the algorithm in the early steps. Any rule to stop the recursion can be used; for example, stop whenever the parameters' change is less than a threshold. In the next subsection, we will show that there is a specific objective function called $ELBO$ that PRISCA maximizes. Algorithm~\ref{alg:prisca_fit} already includes such stopping rule. A difference with additive model is that every update does not give a new parameter estimate, rather a new distribution $p(\bfgamma_l,\bftau_l |\overline{\bfr}_l)$, fully defined by the vectors of probabilities $\bfalpha_l$ and parameters $(a_{t,l},b_{t,l})_{t=1:T}$.


\begin{algorithm}[!b]
	\caption{Fit of PRISCA}
	\label{alg:prisca_fit}
	\begin{algorithmic}
		
		\State \textbf{Inputs:} \bfy, L, $a_0$, $\epsilon$
		\State \textbf{Output:} $(\bfalpha_l)_{1:L}$, $[a_{t,l},b_{t,l}]_{l=1:L,t=1:T}$
		
		\begin{enumerate}
			\item Initialize all entries of  $(\overline{\bftau}_l^2)_{1:L}$ equal to one and compute $ELBO$ via \eqref{eq:elbodecomposed_q_l}
			\item Repeat until convergence 
			
			For $l$ in $1$ to $L$
			
			\begin{itemize}
				\item $\overline{\bfr}_l^2= \bfy^2 \circ \prod_{l'\neq l} \overline{\bftau}_{l'}^2$ 
				\item Update
				$[a_{t,l},b_{t,l}]_{t=1:T}$ via  \eqref{eq:posttau} using $\overline{\bfr}_l^2$ in lieu of  $\bfy^2$
				\item Update $\bfalpha_l$ via \eqref{eq:postalpha} (to compute $P(\bfy | \gamma_t=1,\bfpi, \sigma^2,a_0)$ use (S1) in  the Supplementary Material)
				\item Update  $\overline{\bftau}_l$ via \eqref{eq:expectau}
				\item Compute $ELBO'$ via \eqref{eq:elbodecomposed_q_l}
				\item Exit if $| ELBO'-ELBO| < \epsilon$ is true. If false, set $ELBO=ELBO'$
			\end{itemize}
			
		\end{enumerate}
	\end{algorithmic}
\end{algorithm}

Algorithm~\ref{alg:prisca_fit} outputs $L$ vectors $\bfalpha_l$ and the matrix $[a_{t,l},b_{t,l}]_{l=1:L,t=1:T}$. Each vector can be processed independently to give a point estimate and a credible set. When $L>K$, there are $L-K$ redundant vectors that should not capture any effect, and we expect the corresponding vector $\bfalpha$ to be fairly flat. No matter what $\bfalpha$ looks like, criteria \eqref{eq:point}-\eqref{eq:cs} will return a point estimate and a credible set, even with $\bfalpha$ is as in Figure~\ref{fig:example1} (B) bottom panel. Such behavior is not desirable. We employ a detection threshold to prevent that, classifying a change point as detected if its corresponding credible set is of cardinality less or equal to half of the sequence length. Thus,

\begin{equation}\label{eq:khat}
	\hat{K}=|\{	\mathcal{CS}(\bfalpha_l, p): |\mathcal{CS}(\bfalpha_l, p)|\leq T/2 \}|.
\end{equation}
The threshold $T/2$ is somewhat arbitrary, but our results tend to be entirely insensitive for this threshold as $\bfalpha$ is either very diffuse or very concentrated. 

One could wonder whether an alternative to the $L-K$ redundant vectors being flats is two or more $\alpha_l$ concentrating around the same or a similar time instance. While in theory, this could happen, empirically, we see this happening only in the first algorithm's iterations and then disappear as we iterate the algorithm until convergence via backfitting (see below). Nevertheless, in the implementation, we also include a postprocessing step that removes overlapping credible sets, as done in related works \cite{wang2020simple}. We study this issue in Section~\ref{sec:sim_robu}.

To choose $L$,  one could have prior information over $K$, as in the liver procurement example. In which case, the choice is simple as we can set it equal to that. Lacking prior knowledge, a heuristic is to run Algorithm~\ref{alg:prisca_fit} for multiple values of $L$ starting from one, and stop to increase $L$ when $\hat{K}$ stops rising. We call such a heuristic auto-PRISCA, because it is essentially parameter free since $L$ is not required and the algorithm is not sensitive to $a_0$ as long as it is small. A final possibility is to use the localization rate $\sqrt{T\log T}$ as the minimum spacing condition and set $L$ equal to a number proportional to $T/\sqrt{T\log T}$.

Finally, we notice that each iteration of  Algorithm~\ref{alg:prisca_fit} requires a worst-case cost of $O(TL)$. This is not a problem in practice since the algorithm typically requires a small number of iterations to converge. Thus, for typical datasets Algorithm~\ref{alg:prisca_fit} has an effective overall computational complexity of $O(TL)$.


\subsection{Convergence: Algorithm~\ref{alg:prisca_fit} as Variational Bayes}

Algorithm~\ref{alg:prisca_fit} offers an efficient way to approximate PRISCA's posterior distribution without sacrificing accuracy, as we will show in Section~\ref{sec:sim}. The algorithm somewhat mimics the model's rationale, with multiple single change point models fitted recursively. Although, it is clear that the output offers, at best, an approximation of the true posterior distribution. Given the remarkable empirical performance, a natural question is why Algorithm~\ref{alg:prisca_fit} works and whether there is an underlying model it approximates. In this subsection, we address this question and establish Algorithm~\ref{alg:prisca_fit}'s convergence as a byproduct.  Our result builds on  \cite{wang2020simple}, who show that their recursive algorithm is a Variational Bayes (VB) approximation to the posterior distribution. Despite lacking an additive structure, we will show similar results in our context. 

Let $p$ denote the target posterior distribution. Given that $p$ is seldomly available in closed-form, a popular way to approximate $p$ is VB \cite{wai08}: let  $\mathcal{Q}$ be an arbitrary family of distributions, one chooses  $q\in \mathcal{Q}$ to approximate $p$ with minimal Kullback-Leibler divergence \cite{kul51} between $q$ and $p$, \textit{i.e.}, $q= \argmin_{q \in \mathcal{Q}}KL(q||p)$. If we do not restrict $\mathcal{Q}$ and can solve the resulting optimization problem, we can achieve $KL(q||p)=0$, which corresponds to $q\equiv p$. It means that, under no restrictions on $\mathcal{Q}$, we are simply formulating the problem of calculating the posterior as an optimization problem. VB becomes an approximation when we restrict the class $\mathcal{Q}$.  Such restriction is chosen to make the optimization problem tractable; for example, Mean-Field VB (MFVB) assumes that the variational distribution $q$ factorizes over the variables \cite{wai08}.

Even when imposing restrictions on $\mathcal{Q}$, one seldomly minimizes the KL divergence. The KL between the posterior and the variational family can be rewritten as
\begin{align*}
	KL(q||p)=& \int q(\theta) \log \frac{q(\theta)}{p(\theta|\bfy)}  \di \theta \\ \nonumber
	=&\log p(\bfy)-\int q(\theta)\log \frac{p(\bfy,\theta)}{q(\theta)} \di \theta \nonumber\\
	=& \log p(\bfy) - ELBO(q,\pi,\bfy),
\end{align*}
where $\pi$ denotes the prior distribution on the parameter $\theta$ and the last term is called ``evidence lower bound" (ELBO). 
We are making explicit the dependence of $ELBO(q,\pi,\bfy)$ on the prior distribution $\pi$, which is implicit in $p(\bfy,\theta)$. The $\arg \max_{q \in \mathcal{Q}} ELBO(q,\pi,\bfy)$ is equivalent to the minimizer of $KL(q||p)$ because the marginal log-likelihood $\log p(\bfy)$ does not depend on $q$. The advantage of working with the ELBO is that the optimization can be solved analytically in several instances while the minimization of the KL is not. This is because the intractable normalizing constant $p(\bfy)$ is not involved in the ELBO \cite{bishop2006pattern}.

Let $\mathcal{T}^2$ be a $T\times T$ diagonal matrix with entries $ \bftau^2$, $\mathcal{T}_l^2$, the corresponding matrices with entries $\bftau^2_l$, and $|\cdot|$ the determinant. In our setting, the ELBO can be computed in closed form
\begin{align}
	ELBO(q,\pi ,\bfy):=&\E_q[\log p(\bfy |  \bftau^2_{1:L},\bfgamma_{1:L},\sigma^2)] +\\ \nonumber
	&+ \E_q\bigg[\log \frac{\pi(\bftau^2_{1:L},\bfgamma_{1:L})}{q( \bftau^2_{1:L},\bfgamma_{1:L})}\bigg]\\ \nonumber
	=& - \E_{q}\bigg[\frac{\log |\mathcal{T}^2|}{2}\bigg] -\E_{q}\bigg[\frac{\bfy^T \mathcal{T}^2\bfy}{2 \sigma^2}\bigg]+ \\ &+ \E_{q}\bigg[\log \frac{\pi(\bftau^2_{1:L},\bfgamma_{1:L})}{q(\bftau^2_{1:L},\bfgamma_{1:L})}\bigg] + \text{const.}, \nonumber
\end{align}
where the expected values are computed with respect to $q(\bftau^2_{1:L},\bfgamma_{1:L})$. If we assume that $q$ factorizes over each component through the following mean-field assumption
\begin{equation} \label{eq:mf}
	q(\bftau^2_{1:L},\bfgamma_{1:L})=\prod_{l}q_l(\bftau^2_l,\bfgamma_l),
\end{equation}  
we can further simplify the above as a function of expected values computed with respect to the individual $q_l$
\begin{align}\label{eq:elbodecomposed_q_l}
	ELBO(q,\pi, \bfy)= &  -\sum_l\E_{q_l}\bigg[\frac{\log |\mathcal{T}_{l}^2|}{2}\bigg] -\sum_t \frac{y^2_t \prod_l \E_{q_l}[\tau^2_{t,l}]}{2 \sigma^2}+ \\  \nonumber
	&\sum_l \E_{q_l}\bigg[\log \frac{g_l( \bftau^2_l,\bfgamma_l)}{q_l(\bftau^2_l,\bfgamma_l)}\bigg] +\text{const.},
\end{align}
where $\E_{q_l}[\tau^2_{t,l}]$ can be computed analytically via \eqref{eq:expectau}. The contribution to $ELBO(q,\pi, \bfy)$ relative to the $l$th effect is 
\begin{align}
	ELBO_l(q_l,\pi_l,\bfy):=&-\E_{q_l}\bigg[\frac{\log |\mathcal{T}^2_{l}|}{2}\bigg] -\E_{q_l}\bigg[\frac{\bfy^T \mathcal{T}_l^2\bfy}{2 \sigma^2}\bigg]+\\ \nonumber
	&+ \E_{q_l}\bigg[\log \frac{g_l(\bftau^2_l,\bfgamma_l)}{q_l(\bftau^2_l,\bfgamma_l)}\bigg],
\end{align}
where the second term is $\sum_t y^2_t \E_{q_l}[\tau^2_{t,l}]/2 \sigma^2$. We are now ready to state the key result of the subsection.  

\begin{proposition}\label{prop1}
	Let $\overline{\bfr}_l^2= \bfy^2 \circ \prod_{l'\neq l} \overline{\bftau}_{l'}^2$ be the residuals and $\pi_l$ be the prior distributions on $\bftau^2_l$ and $\bfgamma_l$ given in \eqref{eq:prisca}. Then we have that
	$$\argmax_{q_l} ELBO(q,\pi,\bfy) = \argmax_{q_l} ELBO_l(q_l,\pi_l,\overline{\bfr}_l).$$
\end{proposition}

\begin{proof}
	The dependence on $q_l$ of the first and third terms of $ELBO(q,\pi,\bfy)$ is the same as  that of $ELBO_l(q_l,\pi_l,\overline{\bfr}_l)$. To see that the dependence on $q_l$ of the corresponding terms is also the same, we rewrite it as 
	$$\sum_t \frac{y^2_t \prod_l \E_{q_l}[\tau^2_{t,l}]}{2 \sigma^2}= \sum_t \frac{\overline{r}_{t,l}^2 \E_{q_l}[\tau^2_{t,l}]}{2 \sigma^2}=\E_{q_l}\bigg[\frac{\overline{\bfr}_l^T \mathcal{T}_l^2\overline{\bfr}_l}{2 \sigma^2}\bigg].$$
	The claim then follows.
\end{proof}

\begin{corollary}\label{cor1} The solution to the maximization in Proposition~\ref{prop1} is the posterior distribution of the single change point model defined in \eqref{eq:posterior}-\eqref{eq:postalpha}-\eqref{eq:posttau}, where parameters $\bfalpha_l, (a_{t,l},b_{t,l})_{t=1:T}$ are computed using $\overline{\bfr}_l^2$.
\end{corollary}

Proposition~\ref{prop1} and Corollary~\ref{cor1} establish that Algorithm~\ref{alg:prisca_fit} is a VB approximation to the posterior distribution of PRISCA, where the approximation lies in the fact that we have factorized the $L$ effects using a mean-field approximation. Corollary~\ref{cor1} holds because we did not restrict $q_l$ to belong to a particular family. Hence, the maximizer of $ELBO_l$ is the distribution $q_l$ that makes the $KL$ divergence between $q_l$ and the posterior distribution of a single change point model equal zero. 

Proposition~\ref{prop1} formalizes that each step of Algorithm~\ref{alg:prisca_fit} maximizes the component-wise ELBO. \textit{I.e.}, the procedure to fit PRISCA is a block-wise coordinate ascent. The second result of the subsection follows.


\begin{proposition} Assuming $0<a_0,\sigma^2 <\infty$, the sequence $(\bfalpha_{l,i})_{1:L,i\geq 1}$, $[a_{t,l,i},b_{t,l,i}]_{l=1:L,t=1:T,i\geq 1}$ defined Algorithm~\ref{alg:prisca_fit} converges to a stationary point of $ELBO(q,\pi,\bfy)$. 
\end{proposition}

The convergence to a stationary point follows standard results on block-wise coordinate ascent \cite{bert99,ort00}.   With the condition on the hyperparameters, we require the posterior distribution of each ``single change point problem" we solve to be properly defined. 

In this subsection, we shown that Algorithm~\ref{alg:prisca_fit} estimates the parameters of a variational approximation of the true posterior. The first byproduct of this characterization is the algorithm's convergence. The second one is a closed-form objective function, the ELBO, that Algorithm~\ref{alg:prisca_fit} is maximizing at each iteration. This offers us a formal stopping rule which we will use in the following: stop Algorithm~\ref{alg:prisca_fit} when the growth in the ELBO is less than a prespecified threshold $\epsilon$.

\section{Extensions}\label{sec:ext}

The single change point model consistency allows for dependent data, varying mean, and non-Gaussian data. It is anyway important to extend PRISCA to a range of more realistic and challenging settings to boost  the finite sample performance. The scope is mostly to illustrate how easily our framework can be generalized. We acknowledge the limitations of the algorithms discussed below, and that much research is required to understand their theoretical and empirical properties. This goes beyond the scope of the paper and leave it to future work.

To accommodate a varying mean or an autoregressive process, we introduce additional model parameters, denoted as $\theta$. The likelihood becomes  $p(\bfy | \bftau^2_{1:L},\bfgamma_{1:L},\theta)$ (we omit the baseline variance $\sigma^2$ for parsimony). To mitigate the computational burden of our algorithm, we propose to use an EM-type algorithm, where we iterate between optimizing for $\theta$, and computing the variational approximation $q$. \textit{I.e.} $\theta$ and $q$, and repeat iteratively for $t>1$
\begin{align*}
	\theta^{(t+1)}&=\argmax \E_{q^{(t)}} [\log p(\bfy | \bftau^2_{1:L},\bfgamma_{1:L},\theta)]\\
	q^{(t+1)}& =\argmax ELBO_{\theta^{(t+1)}}(q,\pi,\bfy),
\end{align*}
where the $ELBO$ depends on parameter $\theta^{(t+1)}$ at time $t+1$. A difference with the traditional EM algorithm \cite{dempster1977maximum} is that the distribution $q$ is a variational approximation. Empirical observations suggest that this strategy maintains the computational feasibility of our proposal, albeit at the cost of losing the uncertainty quantification characteristic of Bayesian methods. To retain this uncertainty quantification, we can employ the Variational Bayesian EM algorithm (\textit{e.g.}, \cite{beal2003variational}), which is a straightforward extension of the approach above. This paper primarily focuses on change points, treating $\theta$ as a nuisance parameter. A fully Bayesian approach on $(\bftau^2_{1:L},\bfgamma_{1:L},\theta)$ will be left for future work.

\textbf{\textit{Smoothly changing mean.}} There are plenty of phenomena described by a stochastic process undergoing multiple variance change points in the presence of a smoothly varying mean trend \cite{gao19}. For example, the data considered in this paper exhibits such a pattern (liver temperatures, Figure~\ref{fig:liver}; wave heights, Figure~\ref{fig:ocean}). A possible model is $Y_t \sim N(f_t, \sigma_t^2)$, $t$ in $1$ to $T$, with $f_t$ a smooth function (\textit{e.g.}, assume it belongs to a reproducing kernel Hilbert space) and $\sigma_t^2$ is piecewise constant with $K$ breakpoints and sparse, in the sense that $K<<T$. \cite{gao19} tackle the setting $K=1$ with an algorithm that iterates a penalized weighted least square to estimate $f_t$ and the detection of a single change point with a generalized likelihood ratio. Each step is fed using the residuals computed from the previous one.

Algorithm~\ref{alg:prisca_fit} has a similar recursive structure. Rather than estimating $f_t$ and the change points in two separate steps, we include the estimation of the mean as part of each cycle of the backfitting algorithm \cite{fri81}. This is possible because we wrote our model as it was an additive model. Any algorithm to estimate $f_t$ with heteroskedastic noise can be used. Here, we estimate $f_t$ via a weighted least squares solution of the trend filtering algorithm \cite{kim2009ell_1,tib14}. The vector of weights for the least squares solution is the vector of expected precisions $\mathbf{w}= \prod_{l=1}^L \overline{\bftau}_l^2$, where the expectations $ \overline{\bftau}_l$ are computed with respect to the variational distribution and are used to scale the innovations. The \texttt{R} package \texttt{glmgen} implements this algorithm. We add trend filtering as an extra step in the iterative procedure we use to estimate PRISCA's parameters. The estimate $(\hat{f}_t)_{1:T}$ are used to compute residuals $\overline{\bfr}_l^2$ (see Algorithm~\ref{alg:prisca_fit}): 
\[
\overline{\bfr}_l^2= \left(y_t-\hat{f_t}\right)^2_{t=1:T} \circ \prod_{l'\neq l} \overline{\bftau}_l^2,
\]
which are used in the variance change point detection portion of the algorithm. PRISCA plus the extra step (we will call it TF-PRISCA) outputs estimates of $f_t$, the number and locations of multiple change points, and credible sets. To our knowledge, there is no available method with these features. 

\textbf{\textit{Autoregression.}} The theory in Section~\ref{sec:theory} covers the case of dependent data. Change point detection in the presence of dependent noise has also received considerable attention recently \cite{det20}. Here, we consider an extension of PRISCA that is designed to improve the finite sample performance of the algorithm in the presence of dependent data. Assume  $Y_t = \sum_{i=t-1}^{t-r} \phi_i Y_i +\epsilon_t$, and $\epsilon_t \sim N(0, \sigma_t^2)$, $t$ in $1$ to $T$, with $\sigma_t^2$ piecewise constant with $K$ breakpoints and sparse. For simplicity, assume that the order $r$ is known. The case of known AR coefficients $(\phi_i)_{1:r}$ is uninteresting: one can compute the noise sequence $(\epsilon_t)_{1:T}$ and fit PRISCA directly to it. 

In the case of unknown AR coefficients, PRISCA's extension is similar to the smoothly varying mean case: we add an additional step in the backfitting algorithm to estimate the AR coefficients. A simple solution is to add an extra iteration in Algorithm~\ref{alg:prisca_fit}, which consists of a weighted least squares estimation of $(\phi_i)_{1:r}$. Any method to estimate the coefficients could be employed. As an illustration, we estimate the AR coefficients as the weighted least square' solution of the AR linear model. The estimates $(\hat{\phi}_i )_{i=1:r}$ are then used to compute residuals $\overline{\bfr}_l^2$: 
\[
\overline{\bfr}_l^2= \left(y_t-\sum_{i=1}^r \hat{\phi}_i y_{t-i}\right)^2_{t=1:T} \circ \prod_{l'\neq l} \overline{\bftau}_l^2,
\]
where we set $(y_{-r},\dots,y_0)=\mathbf{0}$. Residuals feed the change point detection loop of Algorithm~\ref{alg:prisca_fit}. There is nothing specific to weighted least squares. As in the previous section, any method to estimate $(\hat{\phi}_i )_{i=1:r}$ in the presence of a time-varying variance can be used. We call this modification AR-PRISCA.

\begin{algorithm}[!b]
	\caption{Fit of TF-PRISCA or AR-PRISCA}
	\label{alg2}
	\begin{algorithmic}
		
		\State \textbf{Inputs:} \bfy, L, $a_0$, $\epsilon$,  (if AR-PRISCA, also $r$)
		\State \textbf{Output:} $(\bfalpha_l)_{1:L}$, $[a_{t,l},b_{t,l}]_{l=1:L,t=1:T}$
		
		\begin{enumerate}
			\item Initialize all entries of  $(\overline{\bftau}_l^2)_{1:L}$ equal to one, $\hat{f_t}=0$ for all $t$, and compute $ELBO$ via \eqref{eq:elbodecomposed_q_l}
			\item Repeat until convergence 
			
			\textit{Step A}. 
			\begin{itemize}
				\item Compute $\mathbf{w}= \prod_{l=1}^L \overline{\bftau}_l^2$
				\item Compute	$\overline{\bfr}^2= \left(y_t-\hat{f_t}\right)^2_{t=1:T} \circ \mathbf{w}$
			\end{itemize}
			If TF-PRISCA
			\begin{itemize}
				\item Compute $(\hat{f_t})^2_{t=1:T}$ via weighted least squares (weights $ \mathbf{w}$) using the trendfiltering (\texttt{glmgen} function in \texttt{R})
			\end{itemize}
			If AR-PRISCA
			\begin{itemize}
				\item Compute $(\hat{\phi}_i )_{i=1:r}$ via weighted least squares (weights $ \mathbf{w}$) of the linear model  
				\item Compute $\hat{f_t}= \sum_{i=1}^r \hat{\phi}_i y_{t-i}$ for all $t$. 
			\end{itemize}
			
			\textit{Step B}. 			
			For $l$ in $1$ to $L$
			
			\begin{itemize}
				\item $\overline{\bfr}_l^2=\left(y_t-\hat{f_t}\right)^2_{t=1:T} \circ \prod_{l'\neq l} \overline{\bftau}_l^2$ 
				\item Update
				$[a_{t,l},b_{t,l}]_{t=1:T}$ via  \eqref{eq:posttau} using $\overline{\bfr}_l^2$ in lieu of  $\bfy^2$
				\item Update $\bfalpha_l$ via \eqref{eq:postalpha} (to compute $P(\bfy | \gamma_t=1,\bfpi, \sigma^2,a_0)$ use (S1) in supplementary material)
				\item Update  $\overline{\bftau}_l$ via \eqref{eq:expectau}
			\end{itemize}
			
			\textit{Step C}. 	
			\begin{itemize}
				\item Compute $ELBO'$ via \eqref{eq:elbodecomposed_q_l}
				\item   Exit if $| ELBO'-ELBO| < \epsilon$ is true. If false, set $ELBO=ELBO'$
			\end{itemize}
			
		\end{enumerate}
	\end{algorithmic}
\end{algorithm}

\textbf{\textit{Misspecified models.}}  We developed the model under Gaussian assumptions and proved that the detection rule \eqref{eq:point} is robust to model mispecification. A growing body of literature tackles model misspecification with power posteriors, posterior distribution where the likelihood is raised to a certain power $\beta$, \textit{i.e.},  $p^{\beta}(\theta| \bfy) \propto \mathcal{L}(\bfy | \theta)^{\beta} \pi(\theta)$ for some power $\beta<1$ (the standard notation for the power is $\alpha$, we employ $\beta$ here as we used  $\bfalpha$ for the posterior distribution over time instances) \cite{grunwald2012safe,holmes2017assigning,miller2018robust}. The corresponding variational approximation -- labeled fractional variational Bayes -- has also been studied \cite{alquier2020concentration}.

Incorporating these ideas into PRISCA is straightforward, as most of the distribution involved in the standard formulation belong to the exponential family as in the example of \cite{bhattacharya2019bayesian}. To derive the power VB equivalent of PRISCA, we start from the single change point model presented in Section~\ref{sec:single} and redo the calculations detailed in Section~\ref{sec:single_details}.
To compute $p^{\beta}(\bfgamma, \bftau| \bfy)$ -- the power posterior equivalent of $p(\bfgamma, \bftau| \bfy)$ -- we have to compute 
\[
p^{\beta}(s^2|\gamma_t=1,\bfy)\propto \left[\prod^T_{i=t} N(y_i | 0, s^{-2} \sigma^2)\right]^{\beta}  \text{Gamma}(s^2|a_0,b_0),
\]
where the above equals  $\text{Gamma}(a^{\beta}_t,b^{\beta}_t)$ with parameters given by
\begin{equation*}\label{SMeq:gamma_frac}
	a^{\beta}_t =a_0 + \beta \, \frac{T-t+1}{2} \,\,\,\, \text{and} \,\,\,\, b^{\beta}_t=a_0+ \beta \,\frac{\bfy_{t:T}^T \bfy_{t:T}}{2 \sigma^2}.
\end{equation*}
Hence, we still have a Gamma posterior with hyperparameters directly including the power $\beta$. 
Similarly,  $p^{\beta}(\bfgamma| \bfy)$ is a Multinomial distribution with parameters $(1,\bfalpha^{\beta})$, where $\bfalpha^{\beta}$ is a vector with entries $\alpha^{\beta}_t \propto p^{\beta}(\gamma_t=1|\bfy)$. To compute this we need to compute the marginal likelihood
\[
\begin{array}{lll}
	p^{\beta}\left(   \bfy |    \gamma_t = 1 \right) &=& \displaystyle   \int  	[p \left(   \bfy |    \gamma_t = 1,s \right)]^{\beta} dP(s^2 |a_0 )\\
	& \propto & \exp\bigg(   -\frac{\beta}{2\sigma^2} \sum_{j=1}^{t-1} y_j^2   \bigg)
	\frac{\Gamma(  a^{\beta}_t  )}{(  b^{\beta}_t )^{ a^{\beta}_t  }}, \nonumber
\end{array}
\]
where $\Sigma$ is defined in Section~\ref{sec:single_details}.  This means that the power posterior of the single change point model $p^{\beta}(\bfgamma, \bftau| \bfy)= \text{Gamma}( a^{\beta}_t , b^{\beta}_t ) \text{Multinomial}(1,\bfalpha^{\beta})$ has the same functional form as the solution of the ``regular" model, and we just have to change the parametrization. 

The single change point model is the building block of PRISCA. Since the power posterior of the single change point model has the same functional form of the ``standard posterior", Algorithm $1$ will remain the same, with the only difference that we will have to account for the new parameters. Similarly, it is trivial to show that it remains a variational approximation, now of the power posterior instead of the regular posterior.

\textbf{\textit{Non-conjugate priors.}}  The building block of our proposal is the single change-point model described in Section~\ref{sec:single}. We introduced it with a Gamma prior on the unknown precision parameter (see (1) in the manuscript), which is the natural choice due to conjugacy. The main advantage of this choice is the possibility to compute the quantities in Section~\ref{sec:single_details} in closed form without resorting to numerical approximation. PRISCA's idea is to heavily rely on solving efficiently single-change point models. A natural question is whether our proposal generalizes to other prior distributions. In hierarchical models, \cite{gelman2006prior} discusses the limitations of the Gamma prior as a ``noninformative prior" and constructs a new folded-noncentral-$t$ family (to be precise, the paper discusses Inverse-Gamma on the variance). 

PRISCA does not require a closed-form expression for the variance posterior distribution. All one needs to apply Algorithm $1$ is the marginal likelihood $p(\bfy  | \gamma_t=1)$ and the posterior expected value $\E [s^2| \bfy,\gamma_t=1]$ of the single change point model for all $t$, which in turns allows us to compute $\E [\tau^2_t| \bfy]$. If we were to choose a different prior on the variance, methods like MCMC would allow us to approximate the posterior, but much of the computational benefits of PRISCA would be lost. A simple alternative is to approximate  $p(\bfy  | \gamma_t=1)$ with a Laplace approximation, while for $\E [s^2| \bfy,\gamma_t=1]$  one can use the MAP estimator, which we need to compute anyway for the Laplace approximation. Alternatively,  there are instances where $\E [s^2| \bfy,\gamma_t=1]$  is available in closed-form even for non-conjugate priors \cite{polson2012half}.

We study whether these approximations introduce a numerical error that makes not viable the use of PRISCA with nonconjugate priors. To separate the effect of the numerical approximation from the choice of the prior, we consider the same Gamma prior on $s^2$, but we solve PRISCA using Laplace approximations for $p(\bfy  | \gamma_t=1)$ and the MAP estimator for $\E [s^2| \bfy,\gamma_t=1]$. Both quantities are available in closed form: $\E [s^2| \bfy,\gamma_t=1] \approx s^{2, MAP}_t=[(T-t)/2 + (a_0-1)]/(\sum_{i=t}^T y_i^2 + b_0)$ and we have that $p(\bfy  | \gamma_t=1) \approx \widehat{p_{G}(\bfy  | \gamma_t=1)}$, where $p_G$ denotes the Laplace approximation, which is equal to 
\[
e^{ -(\sum_{i=t}^T y_i^2 +b_0)s^{2, MAP}_t + \frac{(T-t)/2 + (a_0-1)}{s^{2, MAP}_t }} \frac{\sqrt{2 \pi}}{\sqrt{e^{ \frac{(T-t)/2 + (a_0-1)}{s^{2, MAP}_t}}}}.
\]
We consider a second model with a Half-Cauchy prior on the standard deviation. This is one of the priors discussed by \cite{gelman2006prior}. Recall that if $\sigma$ is Half-Cauchy distributed, it will have density
\[
p(\sigma| \xi) = \frac{2}{\pi \xi [1+(\frac{\sigma}{\xi})^2]},
\]
where $\xi$ can be tuned  to make the prior ``uninformative". Here, the log posterior
\[
\begin{array}{l}
	\log p(\sigma| \bfy, \gamma_t=1,\xi) \propto -\frac{\sum_{i=t}^T y_i^2}{2 \sigma^2} - (T-t) \log \sigma \\
	\,\,\,\,\,\,\,\,\,\,\,\,\,\,\,\,\,\,\,  \,\,\,\,\,\,\,\,\,\,\,\,\,\,\,\,\,\,\,  \,\,\,\,\,\,\,\,\,\,\,\,\,\,\,\,\,\,\,- \log  \left[1+(\frac{\sigma}{\xi})^2\right]\\
\end{array}
\]
can be maximized using any numerical solver. In our implementation, we use the \texttt{R} function \texttt{optim}.

\textbf{\textit{Multiple observations per time point - Periodic data.}} Throughout the paper, we employed the standard assumption that a single sample is available per time point. However, there are situations where practitioners require a model that can accommodate multiple observations per time point. For example, this can happen if the reported data are binned into time intervals.

An interesting case of repeated observations is when measurements are taken cyclically (\textit{e.g.}, at every hour of the day, every day), and the observations process exhibit a periodic structure (\textit{e.g.}, there is a shift in the morning and one at night). There are different proposals to deal with this setting \cite{lund07,ush22}. A possible solution is to consider each cycle as a complete realization of the process, resulting in multiple observations per time point. Here, the cycle length equals $T$, and $n_t$ is the number of samples at time $t$. Bayesian change point methods naturally handle multiple observations collected at a given instance \cite{liu20,cap21}, and PRISCA is no exception.

\section{Simulations}
\label{sec:sim}

\begin{table}\centering 
	\caption{\small{\textbf{Simulation study.} Averages across $300$ repetitions for different sample sizes (T): bias $K-\widehat{K}$ (the lower, the better), Hausdorff statistics $d(\widehat{\mathcal{C}},\mathcal{C}^*)$ (the lower, the better), time (in seconds), average credible set length (length), and coverage of the sets conditional on detection (cond. cov.). PELT, BINSEG, and SEGNEI are point estimators, so it is not possible to give length and cond. cov. PRISCA has $L=\lfloor T/30  \rfloor$. ora-PRISCA has $L=K$, auto-PRISCA $L$ is set automatically. $\mathcal{CS}$ of PRISCA-based methods and MCP are constructed at $p=0.9$.}} 
	\label{tab:sim} 
	\scalebox{0.8}{\begin{tabular}{@{\extracolsep{5pt}} cc|ccc|cc} 
			\\[-1.8ex]\hline 
			
			T & Method & $K-\widehat{K}$ & $d(\widehat{\mathcal{C}},\mathcal{C}^*)$  & Time & Length & Cond. Cov. \\ 
			\hline \\[-1.8ex] 
			200 & MCP & 1.7 & 106.3 & 27.86 & 24.04 & 0.95 \\ 
			& auto-PRISCA & 1.38 & 72.57 & 0.15 & 13.49 & 0.81 \\ 
			& BINSEG & 2.08 & 117.72 & 0 &  &  \\ 
			& PELT & 1.92 & 106.96 & 0 & &  \\ 
			& PRISCA & 1.49 & 79.48 & 0.01 & 13.33 & 0.82 \\ 
			& ora-PRISCA & 1.45 & 77.74 & 0.01 & 14.6 & 0.83 \\ 
			& SEGNEI & 1.63 & 103.94 & 0.21 &  &  \\ 
			\hline
			500 & MCP & 2.64 & 182.12 & 219.8 & 57.37 & 0.99 \\ 
			& auto-PRISCA & 2.03 & 110.89 & 0.61 & 18.99 & 0.82 \\ 
			& BINSEG & 3.05 & 206.9 & 0.01 &  &  \\ 
			& PELT & 2.73 & 175.9 & 0 &  &  \\ 
			& PRISCA & 2.02 & 124.96 & 0.18 & 18.68 & 0.84 \\ 
			& ora-PRISCA & 1.91 & 114.38 & 0.04 & 20.91 & 0.85 \\
			& SEGNEI & 2.41 & 185.84 & 0.34 &  &  \\ 
			\hline
			1000 & MCP & 4.21 & 360.61 & 1830.2 & 93.46 & 0.99 \\ 
			& auto-PRISCA & 2.94 & 186.07 & 2.82 & 23 & 0.84 \\
			& BINSEG & 3.91 & 342.04 & 0.01 &  &  \\ 
			& PELT & 3.39 & 263.76 & 0.01 &  &  \\ 
			& PRISCA & 2.55 & 200.83 & 1.69 & 23.91 & 0.86 \\
			& ora-PRISCA & 2.55 & 181.32 & 0.17 & 26.21 & 0.86 \\ 
			& SEGNEI & 3.22 & 301.04 & 1.1 &  &  \\
			\hline \\[-1.8ex] 
	\end{tabular} }
\end{table}

				%
		
		\subsection{Multiple change points}
		
		To validate PRISCA, we simulate sequences of zero-mean Gaussian observations experiencing multiple changes in variance. The setup of the simulation study is inspired by Killick et al.  \cite{killick2012optimal}. We consider data sets with varying lengths ($T \in \{200,500,1000\}$), varying change point locations, and different variances within each segment. The number of changes is set to $K=\lfloor T^{1/2}/4 \rfloor$. For each simulate data set, we sample new change point locations $\mathcal{C}=\{t_1^*,\ldots,t_K^*\}$  from a uniform distribution on $\{2,\ldots,T-2\}$ with an additional constraint that the minimum spacing between change points is $\min \{T^{1/2},30\}$ ($T^{1/2}$ is justified by the localization rate, see Section \ref{sec:theory}). Variances within each segment ($\sigma^2_1,\ldots,\sigma_{K+1}^2$) are $i.i.d.$ samples from a lognormal distribution with mean $0$, standard deviation $\log(10)/2$.  We generate $300$ data per $T$. 
		
		By design, the simulation study is challenging because variances in consecutive blocks ($\sigma^2_i$ and $\sigma^2_{i+1}$) have positive probability of being practically identical. What matters is the relative performance of PRISCA vis-a-vis state-of-the-art methods. We compare PRISCA's point estimates to PELT \cite{killick2012optimal}, Binary Segmentation (BINSEG) \cite{scot74}, Segment Neighbourhoods (SEGNEI) \cite{aug89} (the three methods are available in the \texttt{R} package \texttt{changepoint}; \cite{killick2014changepoint}). We also include the Bayesian method MCP \cite{mcp}, which efficiently implements the Gibbs sampler of Carlin et al. \cite{carlin1992hierarchical}. MCP is the only method that requires knowledge of the true number of change points. We can obtain both point estimates and credible sets from MCP posterior distributions using the same rules employed by PRISCA.  We also include a comparison to the Bayesian method of Fearnhead \cite{fea06} in the Supplementary Material. We do not add it in the manuscript because, to compare it to the other methods, we had to add a postprocessing step to obtain point estimates and credible sets. This affects the method's performance and makes the comparison unfair.

		To measure accuracy, we use $K-\widehat{K}$ to measure how well each estimator recovers the actual number of change points. We also consider a Hausdorff-like statistic
		$d(\widehat{\mathcal{C}}, \mathcal{C}^*)=\underset{\eta \in \mathcal{C}^*}{\max} \,\ \underset{x \in \widehat{\mathcal{C}}}{\min} \vert x- \eta\vert$ to assess how well each method estimates the true change points locations $t^*_1,\ldots, t^*_K$. To evaluate PRISCA's credible sets, we report the average coverage of each set conditional on detection. For this study, a changepoint is considered to be detected if PRISCA identifies its location within a distance of $\min \{T^{1/2},30\}/2$ of the true position. Such notion detection is used by Killick et al. \cite{killick2012optimal}, and we borrowed it in this study. We also report the average sets' length. These measures are not available for PELT, BINSEG, and SEGNEG because they provide only point estimates, but they are for MCP. 
		

		PRISCA's parameters are set to $a_0=.001$ (\textit{i.e.}, uninformative Gamma prior; results are unaffected by $a_0$, see Section~\ref{sec:sim_robu}), $L=\lfloor T/30 \rfloor$, (\textit{i.e.}, the maximum number of change points is set equal to the minimum spacing condition times the sequence length), and $\epsilon=.001$ (ELBO convergence). We initialize the algorithm as if there were no change points ($(\overline{\bftau}_l^2)_{1:L}$ equal to one), We also report results for auto-PRISCA ($L$ chosen automatically through the heuristic described in  Section~\ref{sec:prisca_algo}) and an ``oracle" version of PRISCA (ora-PRISCA) where we assume knowledge of the correct number of changes, \textit{i.e.}, $L=K$. Table~\ref{tab:sim} summarizes the results. 
		
		As expected by the study construction, no method accurately estimates $K$ ($K-\widehat{K}$ positive on average), but PRISCA-based methods attain the lowest bias across sample sizes. Similarly, PRISCA-based methods report the lowest $d(\widehat{\mathcal{C}}\vert\mathcal{C}^*)$. This suggests that our approach outperforms state-of-the-art methodologies in this particular simulation study. We do not claim that PRISCA is better than state-of-the-art estimators. Rather than it can achieve comparable performance as a point estimator while also providing a measure of uncertainty. Notably, it is orders of magnitude faster than competing Bayesian methodologies we tested. 
		
		Table~\ref{tab:sim} reports the summary measures on the credible sets (here, $p=0.9$). The average length is small, providing evidence that the sets are not dispersed and that the posterior distributions of $\gamma_l$ concentrate on few points. The maximum length sets across all datasets are always well below the $50\%$ threshold we use in our detection threshold. We can see that PRISCA's sets do not attain the targeted $p$ level, but they are relatively close. This behavior is largely expected because VB estimates are known to underestimate uncertainty \cite{bishop2006pattern}. While the underestimation is exceptionally severe in some VB applications, it does not seem to be the case here, especially considering the challenging simulation study. In addition, we see that if we have prior knowledge on the number of change points (ora-PRISCA), the coverage improves.
		
		Compared to other Bayesian methods, PRISCA-based methods have average run-times orders of magnitude smaller than MCP and  Fearnhead \cite{fea06}. In this example, PRISCA has lower bias and better recovery of change point locations. MCP's sets are well above the targeted $p$ level, possibly obtained because the sets are much bigger; \textit{e.g.}, almost four times bigger for $T=1000$.
		
		Finally, PRISCA and auto-PRISCA have practically identical performance across all metrics, with auto-PRISCA having just slightly average run times. Recall that auto-PRISCA is essentially parameter-free since the only relevant parameter ($L$) is set through a heuristic. Surprisingly, the knowledge of $K$ used in ora-PRISCA does not lead to dramatic improvement in performance, suggesting that our heuristics to choose $L$ are a viable replacement of a priori knowledge of $K$.

		\subsection{PRISCA extensions}\label{sec:sim_ext}
		
		In this subsectiion, the data generating mechanism, PRISCA's parameters, and the evaluation criteria used are identical to the previous one if not otherwise specified. 
		
		\noindent \textit{Non-Gaussian noise.} We simulate sequences of zero-mean observations experiencing multiple changes in variance. The set-up of the simulation study is almost identical to the one seen in the previous subsection, except that we do not simulate Gaussian random variables. Let $\sigma^2_t$ denote the variance at time $t$, we consider two scenarios: 1) $Y_t \sim 0.9 N(0,\sigma^2_t )+ 0.1 N(0,4\sigma^2_t)$, \textit{i.e.}, we observe a sequence of random variables that are distributed as a mixture of Gaussians, the first identical to the original simulation study, a second one with inflated variance; 2) $Y_t=\sigma_t T_t$ with $T_t \sim t_5$, \textit{i.e.}, we observe a sequence of Student's $t$-distributed random variables with parameter $\nu=5$. We compare the numerical performance of the algorithms seen up to now (PRISCA, PELT, SEGNEI, BINSEG) to a version of PRISCA that approximate the power posteriors (we consider two $\beta$s, $0.8$ and $0.9$) rather than the traditional posterior (we label the two by the corresponding $\beta$, $0.8$-PRISCA, $0.9$-PRISCA). Table~\ref{tab:robust} reports the results. 
		
		PRISCA remains the best point estimator across scenarios. The data are sampled from a Student's $t$, the advantages of using PRISCA as a point estimator are notable in comparison to PELT, SEGNEI, and BINSEG. This is in line with what has been previously observed in the literature \cite{lai2011simple,cap21}; \cite{lai2011simple} include a lengthy discussion on why this could be the case. However, the coverage of the sets for these misspecified sets drops substantially. The fractional version of PRISCA ($0.8$-PRISCA, $0.9$-PRISCA) seems capable of counterbalancing this loss in coverage, even if the coverage remains below the correctly specified scenario. Notably, the better coverage is not achieved through larger credible sets. Both $0.8$-PRISCA and  $0.9$-PRISCA are less accurate as point estimator than PRISCA, but still remains more accurate than competitors.

		\begin{table}\centering 
			\caption{\small{\textbf{Simulation study: robustness to different noise models.} Averages across $100$ repetitions for different sample sizes (T): bias $K-\widehat{K}$ (the lower, the better), Haussdorff statistics $d(\widehat{\mathcal{C}},\mathcal{C}^*)$ (the lower, the better), time (in seconds), average credible set length (length), and coverage of the sets conditional on detection (cond. cov.). PELT, BINSEG, and SEGNEI are point estimators, so it is not possible to give length and cond. cov. PRISCA is the method presented in the manuscript, $0.8$-PRISCA and $0.9$-PRISCA are the fractional variational approximation of the power posterior with $\beta$ equal to $0.8$ and $0.9$ respectively.}} 
			\label{tab:robust} 
			\scalebox{0.76}{	\begin{tabular}{@{\extracolsep{5pt}} cl|ccccc} 
					\\[-1.8ex]\hline 
					& & 	\multicolumn{5}{c}{Mixture of Gaussians} \\
					T & Method & $K-\widehat{K}$ & $d(\widehat{\mathcal{C}},\mathcal{C}^*)$  & Time & Length & Cond. Cov. \\ 
					\hline \\[-1.8ex] 
					200 & BINSEG & 2.32 & 142.99 & 0.01 &   &      \\ 
					& PELT & 2.04 & 130.21 & 0 &   &   \\ 
					&  0.8-PRISCA & 1.78 & 107.39 & 0.01 & 17.98 & 0.71  \\ 
					&  0.9-PRISCA & 1.62 & 98.23 & 0.01 & 17.24 & 0.7  \\ 
					& PRISCA & 1.49 & 92.35 & 0.02 & 16.12 & 0.66  \\ 
					& SEGNEI & 1.68 & 119.96 & 0.28 &   &     \\ 
					\hline
					500 & BINSEG & 3.6 & 283.22 & 0.01 &   &     \\ 
					& PELT & 2.78 & 222.63 & 0 &   &    \\ 
					&  0.8-PRISCA & 2.68 & 188.71 & 0.12 & 26 & 0.71 \\ 
					&  0.9-PRISCA & 2.43 & 171.3 & 0.14 & 23.78 & 0.69  \\ 
					& PRISCA & 2.12 & 159.41 & 0.17 & 22.26 & 0.65  \\ 
					& SEGNEI & 2.79 & 245.08 & 0.36 &   &   \\ 
					\hline
					1000 & BINSEG & 4.41 & 447.02 & 0.01 &   &   \\ 
					& PELT & 3.12 & 337.37 & 0.01 &   &   \\ 
					&  0.8-PRISCA & 3.43 & 296.5 & 1.22 & 30.43 & 0.71  \\  
					&  0.9-PRISCA & 3.1 & 278.5 & 1.32 & 27.23 & 0.67  \\ 
					& PRISCA & 2.67 & 257.82 & 1.55 & 23.17 & 0.62 \\ 
					& SEGNEI & 3.49 & 392.84 & 1.14 &   &    \\ 
					\hline \\[-1.8ex] 
					
					& & 	\multicolumn{5}{c}{Student's $t$ ($\nu=5$)}\\
					\hline \\[-1.8ex] 
					200 & BINSEG   & 1.89 & 111.38 & 0.01 &   &   \\ 
					& PELT    & 1.7 & 100.56 & 0 &   &   \\ 
					&  0.8-PRISCA & 1.31 & 81 & 0.02 & 14.59 & 0.7 \\ 
					&  0.9-PRISCA & 1.17 & 74.67 & 0.02 & 13.45 & 0.65 \\ 
					& PRISCA  0.97 & 67.39 & 0.02 & 12.93 & 0.62 \\ 
					& SEGNEI   & 1.21 & 93.32 & 0.26 &   &   \\ 
					\hline
					500 & BINSEG  & 2.7 & 179.68 & 0.01 &   &   \\ 
					& PELT   & 2.33 & 152.68 & 0 &   &   \\ 
					&  0.8-PRISCA  & 1.83 & 122.25 & 0.19 & 19.3 & 0.73 \\ 
					&  0.9-PRISCA  & 1.55 & 112.24 & 0.22 & 17.76 & 0.67 \\ 
					& PRISCA  & 1.29 & 105.56 & 0.25 & 16.58 & 0.63 \\ 
					& SEGNEI   & 1.92 & 160.98 & 0.38 &   &   \\ 
					\hline
					1000 & BINSEG    & 3.31 & 290.8 & 0.01 &   &   \\ 
					& PELT  & 2.76 & 241.07 & 0.01 &   &   \\ 
					&  0.8-PRISCA  & 2.32 & 203.11 & 1.6 & 24.79 & 0.7 \\  
					&  0.9-PRISCA  & 1.9 & 188.72 & 1.76 & 22.77 & 0.65 \\ 
					& PRISCA  & 1.45 & 172.1 & 1.96 & 22.04 & 0.6 \\ 
					& SEGNEI    & 2.43 & 258.28 & 1.1 &   &  \\ 
					\hline \\[-1.8ex] 
			\end{tabular} }
		\end{table}  
		
		\noindent \textit{Varying mean.} We simulate from a noise process with $K=4$,  $\mathcal{C}=\{.15 \,T,.4 \,T,.75 \,T,.85 \,T\}$, and $\sigma_t$ taking respectively values $1,2,3,0.6,$ and $2$ in the five segments defined by  $\mathcal{C}$. The mean is equal to $f_t=20+12t/T(1-t/T)$ as in \cite{gao19}. We consider three $T$s ($200,500,1000$), and $100$ new simulated data sets per combination. We compare TF-PRISCA to an oracle version of PRISCA (ora-PRISCA) that ``knows" the true mean $f_t$. The oracle is essentially the method for independent ordered observations applied to the sequence $(\epsilon_t)_{t=1:T}$. We also include a version to PRISCA that ignores the shift in mean; hence, it applied directly to $(y_t)_{t=1:T}$. Finally, we include a setting where we preprocess the data, \textit{i.e.}, we estimate $(\hat{f}'_t)_{1:T}$ through trend filtering, ignoring heteroskedasticity, and then apply PRISCA to $(y_t - \hat{f}'_t)_{1:T}$ (we call this pre-PRISCA). The latter allows us to test whether there is any advantage in doing the recursive algorithm with backfitting. The algorithms' parameters are identical ($L=8, a_0=.1)$, and we set $k=2$ in the trend filter. There is no competing method for this task.
		
		Table~\ref{tab:simTF} summarizes the results. As a point estimator, the performance of TF-PRISCA is identical to that of the oracle, which corresponds to the accuracy of PRISCA when there is no varying mean. The coverage of the credible sets is somewhat lower. We see that for $T=200$ the coverage is well below $90\%$. As $T$ grows, the coverage stabilizes around $80\%$. 	The additional iteration in TF-PRISCA leads to a substantial increase in the average run time. The growth cannot be attributable to the trend filter, which has an efficient implementation; instead, it comes from the extra iterations required for the ELBO to stop increasing. 
		
		\begin{table}\centering 
			\caption{\small{\textbf{Simulation study: varying mean.} Averages across $100$ repetitions for different sample sizes (T) and mean $f_t$: bias $K-\widehat{K}$ (the lower, the better), Haussdorff statistics $d(\widehat{\mathcal{C}},\mathcal{C}^*)$ (the lower, the better), time (in seconds), average credible set length (length), and  coverage of the sets conditional on detection (cond. cov.). ora-PRISCA knows $f_t$, PRISCA ignores the varying mean, TF-PRISCA estimates $(\hat{f_t})_{1:T}$ with a trend filtering, pre-PRISCA employs a data preprocessing technique. }}
			\label{tab:simTF} 
			\scalebox{0.8}{\begin{tabular}{@{\extracolsep{5pt}} cl|ccc|cc} 
					\\[-1.8ex]\hline 
					
					T & Method & $K-\widehat{K}$ & $d(\widehat{\mathcal{C}},\mathcal{C}^*)$  & Time & Length & Cond. Cov. \\ 
					\hline \\[-1.8ex] 
					200 & PRISCA & 3 & 169 & 0.11 & 1 & 0 \\ 
					& ora-PRISCA & 0.4 & 32.42 & 0.04 & 8.28 & 0.73 \\ 
					& pre-PRISCA & 3 & 169 & 0.01 & 1 & 0 \\ 
					& TF-PRISCA & 0.16 & 30.44 & 4.09 & 8.4 & 0.63 \\ 
					\hline
					500 & PRISCA & 3 & 424 & 0.11 & 1 & 0 \\ 
					& ora-PRISCA & -0.29 & 12.14 & 0.11 & 14.85 & 0.8 \\ 
					& pre-PRISCA & 3 & 424 & 0.01 & 1 & 0 \\ 
					& TF-PRISCA & -0.35 & 15.61 & 2.45 & 14.01 & 0.74 \\ 
					\hline
					1000 & PRISCA & 3 & 849 & 0.13 & 1 & 0 \\ 
					& ora-PRISCA & -0.34 & 9.74 & 0.22 & 16.72 & 0.81 \\ 
					& pre-PRISCA & 3 & 849 & 0.03 & 1 & 0 \\ 
					& TF-PRISCA & -0.38 & 9.5 & 5.02 & 15.69 & 0.78 \\ 
					
					\hline \\[-1.8ex] 
			\end{tabular} }
		\end{table}

		\noindent \textit{Autoregressive noise.} We consider an AR process and test AR-PRISCA. We set $r=1$ and use the same noise process of the previous section ($K=4$,  $\mathcal{C}=\{.15 \,T,.4 \,T,.75 \,T,.85 \,T\}$,  $\sigma_t \in \{1,2,3,0.6,2\}$). We consider three $T$s ($200,500,1000$) and $\phi$s ($0.4,0.6,0.8$) and $100$ new simulated data sets per combination. We compare  AR-PRISCA to an oracle version of PRISCA (ora-PRISCA) that ``knows" the true parameter $\phi$. The oracle is essentially the method for independent ordered observations applied to the sequence $(\epsilon_t)_{t=1:T}$. We also include a version to PRISCA that ignores the dependence; hence, applied directly to $(y_t)_{t=1:T}$. Lastly, we assume that the order $r=1$ is known. For unknown $r$, one can resort to standard times-series techniques, such as selecting the order through information criteria.

		Table~\ref{tab:simAR} summarizes the result. Notably, the performance of AR-PRISCA is practically identical to ora-PRISCA across all metrics (except time), suggesting that the method can correctly account for autoregression. PRISCA gives a snapshot of the performance of an approach that ignores autoregression. As a point estimator, the method remains relatively robust for sample sizes and small $\phi$, with an empirical performance not too far off from the other two. The point estimates become substantially less accurate at $T$ and $\phi$ increase. Ignoring the autoregression leads to inaccurate credible sets with low coverage. Interestingly, the average run-time of AR-PRISCA is not overly inflated by the extra step. AR-PRISCA run-time is sometimes even lower than PRISCA's because it is more challenging to maximize the ELBO when the model is misspecified.

		\begin{table} \centering 
			\caption{\small{\textbf{Simulation study: AR process.} Averages across $100$ repetitions for different sample sizes (T) and autoregressive coefficient ($\phi$): bias $K-\widehat{K}$ (the lower, the better), Hausdorff statistics $d(\widehat{\mathcal{C}},\mathcal{C}^*)$ (the lower, the better), time (in seconds), average credible set length (length), and coverage of the sets conditional on detection (cond. cov.). ora-PRISCA knows $\phi$, PRISCA ignores the AR process, AR-PRISCA estimates $\hat{\phi}$.}}
			\label{tab:simAR} 
			\scalebox{0.78}{\begin{tabular}{@{\extracolsep{5pt}} lcl|ccc|cc} 
					\\[-1.8ex]\hline 
					
					$\phi$ 	& T & Method & $K-\widehat{K}$ & $d(\widehat{\mathcal{C}},\mathcal{C}^*)$  & Time & Length & Cond. Cov. \\ 
					\hline \\[-1.8ex] 
					$0.4$ & 200 & AR-PRISCA & 0.44 & 33.71 & 0.07 & 8.27 & 0.74 \\ 
					&  & PRISCA & 0.38 & 35.22 & 0.15 & 8.18 & 0.7 \\ 
					&  & ora-PRISCA & 0.41 & 32.1 & 0.04 & 8.32 & 0.73 \\ 
					& 500 & AR-PRISCA & -0.25 & 11.39 & 0.14 & 14.79 & 0.81 \\ 
					&  & PRISCA & -0.39 & 22.53 & 0.22 & 14.69 & 0.72 \\ 
					&  & ora-PRISCA & -0.29 & 12.14 & 0.1 & 14.85 & 0.8 \\ 
					& 1000 & AR-PRISCA & -0.32 & 10.03 & 0.27 & 16.68 & 0.81 \\ 
					& & PRISCA & -0.61 & 13.56 & 0.32 & 17.7 & 0.73 \\ 
					&  & ora-PRISCA & -0.34 & 9.74 & 0.2 & 16.72 & 0.81 \\ 
					\hline
					$0.6$ & 200 & AR-PRISCA & 0.39 & 34.11 & 0.07 & 8.2 & 0.72 \\ 
					&  & PRISCA & 0.11 & 34.86 & 0.15 & 8.55 & 0.61 \\ 
					&  & ora-PRISCA & 0.41 & 32.6 & 0.04 & 8.29 & 0.72 \\ 
					& 500  & AR-PRISCA & -0.27 & 12.79 & 0.15 & 14.52 & 0.8 \\ 
					& & PRISCA & -0.87 & 30.66 & 0.22 & 14.08 & 0.53 \\ 
					& & ora-PRISCA & -0.29 & 12.14 & 0.1 & 14.85 & 0.8 \\ 
					&1000 & AR-PRISCA & -0.35 & 10.19 & 0.27 & 16.77 & 0.8 \\ 
					& & PRISCA & -1.25 & 31.13 & 0.32 & 18.31 & 0.57 \\ 
					&  & ora-PRISCA & -0.34 & 9.74 & 0.2 & 16.72 & 0.81 \\ 
					\hline
					$0.8$  & 200 & AR-PRISCA & 0.43 & 34.06 & 0.07 & 8.21 & 0.73 \\ 
					&  & PRISCA & -0.09 & 36.34 & 0.15 & 7.87 & 0.45 \\ 
					&  & ora-PRISCA & 0.41 & 32.36 & 0.04 & 8.32 & 0.73 \\ 
					& 500 & AR-PRISCA & -0.25 & 12.88 & 0.15 & 14.4 & 0.81 \\ 
					&  & PRISCA & -1.82 & 49.78 & 0.22 & 11.08 & 0.3 \\ 
					&  & ora-PRISCA & -0.29 & 12.14 & 0.1 & 14.85 & 0.8 \\ 
					& 1000 & AR-PRISCA & -0.37 & 10.42 & 0.27 & 16.77 & 0.8 \\ 
					&  & PRISCA & -2.24 & 71.95 & 0.33 & 14.96 & 0.3 \\ 
					&  & ora-PRISCA & -0.34 & 9.74 & 0.2 & 16.72 & 0.81 \\ 
					\hline \\[-1.8ex] 
			\end{tabular} }
		\end{table}

		\begin{figure}[!t]
			\begin{center}
				\includegraphics[scale=0.5]{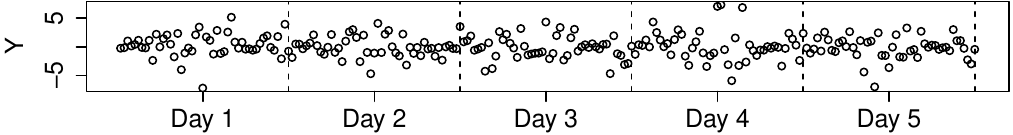}
				\includegraphics[width=.4\textwidth]{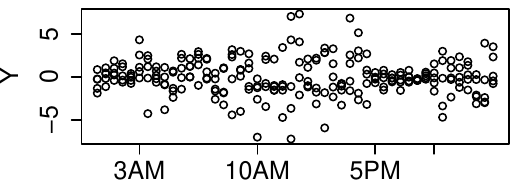}
			\end{center}
			\caption{\small{\textbf{Simulation study: periodic data.}. Top: observations are simulated every 30 mins for five days. There are four daily change points: 3AM, 10AM,5PM, and 8:30PM. Bottom: same as top panel. Here, they are not sorted chronologically but plotted according to the time they are collected.}}
			\label{fig:cycle1}
		\end{figure}

		\noindent \textit{Periodic data.} PRISCA can naturally handle multiple observations at a given time point. In Section~\ref{sec:ext}, we argued that this flexibility is essential because it allows PRISCA to handle several realistic scenarios. Among others, the case of cyclical data with periodic change points. For example, Figure~\ref{fig:cycle1} top panel depicts a simulated data set mimicking observations recorded every half hour for five days ($T=240$). There are four daily change points, at 3AM, 10AM, 5PM, and 8:30PM, but none is visible. Figure~\ref{fig:cycle1} bottom panel depicts the same data set, now plotted according to the time at which the observations are simulated/collected. The change points are now clearly visible. PRISCA naturally handles this data, treating observations collected at the same time as $i.i.d.$.

		We simulate cyclical data as those in Figure~\ref{fig:cycle1}. We have  $K=4$ and $\sigma_t$ taking values $1,2,3,0.6,$ and $2$ in the five segments defined by  the times (same as previous sections). We consider four  number of days collected $D$s ($3,5,10,20$), and $100$ simulated data sets each. Table~\ref{tab:simcycle} summarizes the results. Accuracy increases with sample size (each new day is $48$ more observations). With $20$ days, PRISCA nearly perfectly recovers all the change points. Interestingly, also the credible sets have perfect coverage in this simulation study.

		\noindent \textit{Non-conjugate prior.}
		We replicate the main simulation study and consider different prior distributions on the variance parameter. Table~\ref{tab:sim_nonconju} reports the results. The results on numerical accuracy are essentially identical across the three proposals. The identical performance of PRISCA-Gam-Num and PRISCA suggests that the two approximations introduced (marginal likelihood, posterior expectation) do not substantially affect the performance. Comparing PRISCA to PRISCA-Half-Cauchy ($\xi=25$ as in \cite{gelman2006prior}) suggests that, in this specific setting, the prior on the variance does not seem to have an effect. 
		
		\begin{table}\centering 
			\caption{\small{\textbf{Simulation study: non-conjugate priors.} Averages across $100$ repetitions for different sample sizes (T) and mean $f_t$: bias $K-\widehat{K}$ (the lower, the better), Haussdorff statistics $d(\widehat{\mathcal{C}},\mathcal{C}^*)$ (the lower, the better), time (in seconds), average credible set length (length), and coverage of the sets conditional on detection (cond. cov.). PRISCA is the baseline proposal (Algorithm 1), PRISCA-Gam-Num has a Gamma prior on the precision but approximates the marginal likelihood,  PRISCA-Half-Cauchy has a Half-Cauchy prior on the standard deviation ($\xi=25$).} }
			\label{tab:sim_nonconju} 
			\scalebox{0.75}{\begin{tabular}{@{\extracolsep{5pt}} cl|ccc|cc} 
					\\[-1.8ex]\hline 
					
					T & Method & $K-\widehat{K}$ & $d(\widehat{\mathcal{C}},\mathcal{C}^*)$  & Time & Length & Cond. Cov. \\ 
					\hline \\[-1.8ex] 
					200 & PRISCA & 1.49 & 79.48 & 0.01 & 13.33 & 0.82 \\ 
					& PRISCA-Gam-Num & 1.47 & 78.65 & 0.02 & 13.38 & 0.82 \\ 
					& PRISCA-Half-Cauchy & 1.49 & 80.47 & 1.11 & 13.19 & 0.82 \\ 
					\hline
					500 & PRISCA & 2.02 & 124.96 & 0.19 & 18.68 & 0.84 \\ 
					& PRISCA-Gam-Num & 2 & 125.22 & 0.2 & 18.98 & 0.83 \\ 
					& PRISCA-HalfCauchy & 2.02 & 125.11 & 15.33 & 19.03 & 0.84 \\ 
					\hline
					1000 & PRISCA & 2.55 & 200.83 & 1.86 & 23.91 & 0.86 \\ 
					& PRISCA-Gam-Num & 2.55 & 202.18 & 1.85 & 24.32 & 0.86 \\ 
					& PRISCA-Half-Cauchy & 2.53 & 200.4 & 113.82 & 23.93 & 0.85 \\ 
					\hline \\[-1.8ex] 
			\end{tabular} }
		\end{table}

		\begin{table} \centering 
			\caption{\small{\textbf{Simulation study: periodic data.} Averages across $100$ repetitions for different numbers of days (D): bias $K-\widehat{K}$ (the lower, the better), Hausdorff statistics $d(\widehat{\mathcal{C}},\mathcal{C}^*)$ (the lower, the better), time (in seconds), average credible set length (length),  and  coverage of the sets conditional on detection (cond. cov.). }}
			\label{tab:simcycle} 
			\scalebox{0.8}{\begin{tabular}{@{\extracolsep{5pt}} c|c cc|cc} 
					\\[-1.8ex]\hline 
					$\#$ days	&  $K-\widehat{K}$ & $d(\widehat{\mathcal{C}},\mathcal{C}^*)$  & Time & Length & Cond. Cov. \\ 
					\hline \\[-1.8ex] 
					3 & 0.76 & 10.45 & 0.15 & 2.8  & 0.93 \\ 5 & 0.23 & 5.99 & 0.15 & 2.24  & 0.94 \\ 10 & -0.24 & 1.11 & 0.16 & 1.89  & 0.92 \\ 20 & -0.23 & 0.44 & 0.17 & 1.37  & 0.94 \\ 
					\hline \\[-1.8ex] 
			\end{tabular} }
		\end{table}  
		
		\subsection{Robustness check}\label{sec:sim_robu}

		\noindent \textit{Comparison $p(\gamma_l,\tau_l |\overline{\bfr}_l)$ and $p(\gamma_l,\tau_l |\bfy)$.} PRISCA is a recursion that approximates the posterior distribution via a set of solutions to single-change point models fed with the model residuals $\overline{\bfr}_l$ rather than the true observations. As a heuristics, building the algorithm, we hint that $p(\gamma_l,\tau_l |\overline{\bfr}_l)$ could be approximating the actual marginal posterior distribution $p(\gamma_l,\tau_l |\bfy)$. In the case $L=1$, the two quantities are indeed identical. 
		
		A reviewer suggested justifying the approximation in the single change point scenario, where the exact posterior can be computed in closed form (see Section~\ref{sec:single}). We simulate $100$ data sets with $T=200$, $K=1$, $\sigma_l=1$ and for each data set with sample $\sigma_r\sim \mathcal{U}(0.3,1.6)$ and $t_0 \sim \mathcal{U}_d\{50,\ldots, 150\}$, with $\mathcal{U}_d$ denoting the uniform discrete on a given set. For each dataset, we compute the exact posterior distribution and the corresponding $\hat{t}$ and $	\mathcal{CS}(\bfalpha, p=0.9)$. We also fit PRISCA with a varying number of model blocks  ($L=2$, $L=3$, and $L=4$) (the other parameters as previously specified). We also compute $\hat{K}$ and the corresponding change points estimates and credible sets. 
		
		For all datasets, PRISCA gives the same solution as the single change point model. Whenever a change point is detected by the exact model, the credible sets and the point estimates are identical between the exact model and the three PRISCA fits. In addition we compute the discrete KL divergence between $p(\gamma_l,\tau_l |\bfy)$ and $p(\gamma_l,\tau_l |\overline{r}_l)$  for the $l$th component that detects a change points. Again, the KL is equal to zero for all datasets. This simulation study has clear limitations, starting with the fact that $K=1$. As $K$ increases, the approximation error is likely going to increase. That is also fully understood by the fact that we are doing a variational approximation of the true model posterior. However, this simulation study provides further support for the good empirical performance we observed in the simulation study. 
		
		\noindent	\textit{Overlapping sets.} In Section~\ref{sec:prisca_algo}, we claimed that, while it is theoretically possible that two vectors $\bfalpha_{l_1}$ and $\bfalpha_{l_2}$ capture the same time instance, we observe such behavior only within the first iteration of the algorithm, to then disappear as we run the procedure until convergence. Regardless, we included a postprocessing step in the \texttt{R} package PRISCA that eliminates a set (and the related point estimate) if it overlaps with another one. This simulation study provides empirical proof that this step is rarely needed if the algorithm runs until convergence. We repeat the simulation study in Section~\ref{sec:sim}, comparing the PRISCA's performance with (PRISCA) or without (PRISCA-NP) the postprocessing step. 
		
		Table~\ref{tab:overlap} reports only results ($\#$ estimated change points $\hat{K}$, Haussdorff statistics $d(\widehat{\mathcal{C}},\mathcal{C}^*)$) for three stopping rule $\epsilon$ of the ELBO's convergence. We include only $T=200$ as the pattern is identical across the three sample sizes. As one can see, as the convergence criterion decreases, $\hat{K}$ becomes quickly identical for the two algorithms. Change points locations recovery ($d(\widehat{\mathcal{C}},\mathcal{C}^*)$) is always the same. This is important because it rules out the possibility that the sets defined by $\bfalpha_{l_1}$ and $\bfalpha_{l_2}$ could be slightly different.
		
		\begin{table}\centering 
			\caption{\small{\textbf{Simulation study: overlapping sets.} Averages across $300$ repetitions for $T=200$ for three ELBO's stopping rule ($\epsilon$): estimated number of change points $\widehat{K}$ (the lower, the better), Haussdorff statistics $d(\widehat{\mathcal{C}},\mathcal{C}^*)$ (the lower, the better), time (in seconds). PRISCA removes overlapping sets, PRISCA-NO does not do postprocessing.} }
			\label{tab:overlap} 
			\scalebox{0.8}{\begin{tabular}{@{\extracolsep{5pt}} cl|cc} 
					\\[-1.8ex]\hline 
					
					$\epsilon$ & Method & $\widehat{K}$ & $d(\widehat{\mathcal{C}},\mathcal{C}^*)$   \\ 
					\hline \\[-1.8ex] 
					$10^{-1}$ & PRISCA & 1.51 & 79.48 \\ 
					& PRISCA-NO & 1.71 & 79.2 \\ 
					$10^{-3}$ & PRISCA & 1.63 & 73.57 \\ 
					& PRISCA-NO & 1.66 & 73.49 \\ 
					$10^{-5}$  & PRISCA & 1.65 & 72.49 \\ 
					& PRISCA-NO & 1.65 & 72.49 \\ 
					
					\hline \\[-1.8ex] 
			\end{tabular} }
		\end{table}

		\noindent \textit{Sensitivity to parameter $a_0$.} Table~\ref{tab:a0} summarizes the results of the simulation study where we run PRISCA for various value of $a_0$($10^{-1}$,$10^{-3}$,$10^{-5}$,$10^{-9}$). The results are completely unaffected by the choice of this parameter irrespective of the sample size. 
		
		\begin{table} \centering 
			\caption{\small{\textbf{Simulation study: sensitivity to $a_0$.} Averages across $300$ repetitions for different sample sizes (T): bias $K-\widehat{K}$ (the lower, the better), Hausdorff statistics $d(\widehat{\mathcal{C}},\mathcal{C}^*)$ (the lower, the better), time (in seconds), average credible set length (length), and coverage of the sets conditional on detection (cond. cov.). The algorithm is always PRISCA, run for different level of $a_0$.}}
			\label{tab:a0} 
			\scalebox{0.8}{\begin{tabular}{@{\extracolsep{5pt}} cc|ccc|cc} 
					\\[-1.8ex]\hline 
					
					T & $a_0$ & $K-\widehat{K}$ & $d(\widehat{\mathcal{C}},\mathcal{C}^*)$  & Time & Length & Cond. Cov. \\ 
					\hline \\[-1.8ex] 
					200 & $10^{-3}$ & 1.49 & 79.48 & 0.01 & 13.34 & 0.85 \\ 
					& $10^{-1}$ & 1.51 & 80.42 & 0.01 & 13.29 & 0.85 \\ 
					& $10^{-5}$ & 1.49 & 79.48 & 0.01 & 13.33 & 0.85 \\ 
					& $10^{-9}$ & 1.49 & 79.48 & 0.01 & 13.33 & 0.85 \\ 
					\hline
					500 & $10^{-3}$& 2.02 & 124.96 & 0.17 & 18.75 & 0.87 \\ 
					& $10^{-1}$& 2.03 & 124.97 & 0.17 & 19.1 & 0.88 \\ 
					& $10^{-5}$ & 2.02 & 124.96 & 0.17 & 18.68 & 0.87 \\ 
					& $10^{-9}$ & 2.02 & 124.96 & 0.17 & 18.68 & 0.87 \\ 
					\hline
					1000 & $10^{-3}$ & 2.55 & 201.51 & 1.67 & 23.92 & 0.89 \\ 
					& $10^{-1}$ & 2.59 & 202.24 & 1.63 & 23.83 & 0.89 \\ 
					& $10^{-5}$& 2.55 & 200.83 & 1.67 & 23.91 & 0.88 \\ 
					& $10^{-9}$ & 2.55 & 200.83 & 1.68 & 23.91 & 0.88 \\ 
			\end{tabular} }
		\end{table}

		\section{Real data}\label{sec:real}
		
		\subsection{Liver viability assessment} \label{sec:liver}
		
		
		\begin{figure}[!t]
			\begin{center}
				\includegraphics[scale=0.5]{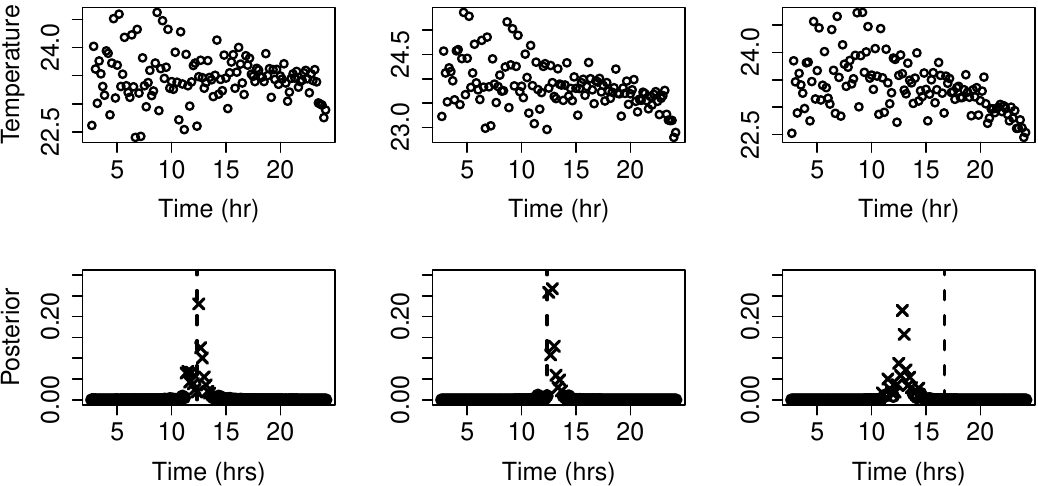}
			\end{center}
			\caption{\small{\textbf{Liver procurement data.} Top row: liver temperature profile at three random locations (temperature every $10$ min. for $24$hrs, first $2.5$hrs not included, $T=130$). Bottom row: posterior distribution of the change locations. Dots depicts the probability of a change point at that location. Crosses denote the points included in a $p=90\%$ credible set. Dashed lines correspond to point estimates of \cite{gao19}' method. }}
			\label{fig:liver}
		\end{figure}

		It has been observed that donor risk is increasing, with a majority of organs classified as marginal. Viability assessment in organ transplantation is a challenging task of utmost importance \cite{panconesi2021viability}. 
		Current techniques are either invasive, \textit{e.g.}, a biopsy of the organ, or highly subjective, \textit{e.g.}, they rely on the physician's judgment of the candidate organs. The former strategy runs the risk of ruining the organ, and the latter is too subjective, given the high number of parameters that must be monitored to assess viability. Much research focuses on replacing these two approaches. Machine perfusion technology have emerged as an alternative in solid organ transplantation and it is an area currently subject of research efforts \cite{friend2020strategies}.

		Here, we analyze experimental data studying a novel dynamic organ preservation strategy, which consists in monitoring the whole organ surface temperature of a liver perfused with a physiologic perfusion fluid (modified Krebs' solution) \cite{gao19}. The experiment records the surface temperatures of a lobe of a porcine liver at $36795$ spots every 10 minutes for $24$ hours. Each temperature profile consists of $T=130$ points because the first $2.5$ hours of data were discarded (it takes about two hours for the perfusion fluid to infuse and stabilize the liver completely). We have data for $100$ random locations. Figure~\ref{fig:liver} top panels depict temperature profiles at three random locations of the liver. A visual inspection suggests a time-varying mean, domain knowledge indicates one change in variance.  The method in  \cite{gao19} was designed for this task.

		To account for the trend, we employ TF-PRISCA (the extension of PRISCA including the trend filter to account for a change in means). We could similarly account for the trend simply by using the first-order difference of the temperatures. Results do not change.  The choice is motivated by the fact that we compare our results with the method from \cite{gao19}. PRISCA's parameters were set to $a_0=.001$, $\epsilon=10^{-5}$, and $L=2$ (one effect for the unknown baseline variance and a second one for the single change point). For the mean, we used the \texttt{glmgen} implementation of the trend filter \cite{tib14}, with $k=2$, we employed a weighted least square solution, and $\lambda$ chosen via BIC. Figure~\ref{fig:liver} bottom panels depict the posterior distribution of $\bfgamma$ at three locations obtained with TF-PRISCA along with the estimates of \cite{gao19} methodology. Crosses depict $p=0.9$ credible sets, dashed lines the point estimates given by \cite{gao19}. Change points from $t=1$ to $t=5$ were excluded from the plots because they refer to the baseline $\sigma^2$ (see also Figure~\ref{fig:example_multi}). 
		
		In two locations, \cite{gao19}' points estimates are in TF-PRISCA credible sets. TF-PRISCA point estimates are off by one instance in these two cases. In the third location (last column), the point estimate of \cite{gao19} is not included in the credible set. A visual inspection supports a change in variance in both locations. We note that if we were to run PRISCA with an extra effect (\textit{i.e.}, $L=3$), we would also recover this second change point. However, the experimental design suggests the existence of a single change. 
		
		The method in \cite{gao19} applied to the $100$ available locations estimate the change points to happen around around $12$ hours $84\%$ of the time, with the remaining $16\%$ estimated at $t=76$. The rightmost panel depicts an example of one of the instances where the change is estimated around $14$ hours. In the latter locations, our method points out that there is evidence for a loss of viability at around $12$ hours as well, hinting that loss of viability was uniform in the different locations.

		\subsection{Oceanographic data}\label{sec:oce}
		
		\begin{figure}[!t]
			\begin{center}
				\includegraphics[scale=0.5]{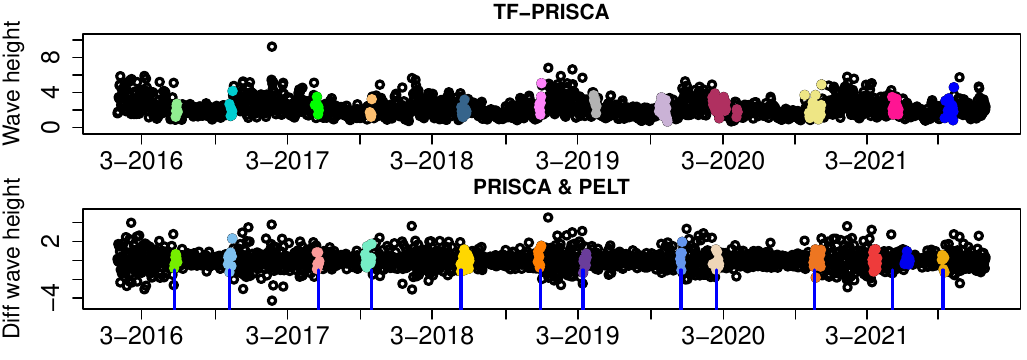}
			\end{center}
			\caption{\small{\textbf{Oceanographic data.} Top panel: daily ($T=2164)$ wave height (ft) recorded by buoy 46042 Monterey (NOOA data) in 2016-2021. Colors depict the time locations in the credible sets ($p=0.9$) constructed by TF-PRISCA. Bottom panel: first-order difference of data in the top panel. Colors depict the time locations in the credible sets ($p=0.9$) constructed by PRISCA. Vertical blue segments depict PELT's point estimates.}}
			\label{fig:ocean}
		\end{figure}
		
		Practitioners who plan maintenance at offshore infrastructures, such as oil rigs and wind farms, study wave height volatility. A low wave height volatility suggests stable sea conditions, which is necessary to minimize risks when organizing repairs and inspections offshore. The application of change point detection methods for this task has been studied in  \cite{killick2010detection}.
		
		We consider wave height data collected hourly from January 2016 to December 2021 by Station 46042 Monterey, a buoy located twenty-seven miles off the coast of Monterey in the Pacific Ocean. The National Oceanic and Atmospheric Administration  National Data Buoy center makes the data publicly available. Figure~\ref{fig:ocean} depicts the sequence of observations subsampled such that we have one measure a day. As noted by \cite{kil12}, the cyclical (seasonal) nature of wave heights is clear, what's less obvious is the location of the changes in volatility, which is the main object of interest \cite{killick2010detection}. We apply PRISCA to estimate variance change points to first-order difference as done by Killick et al. \cite{killick2012optimal}. While first-order difference do a reasonable job in removing the seasonal trend (Figure~\ref{fig:ocean} bottom panel), we also apply TF-PRISCA to the raw data.
		
		Figure~\ref{fig:ocean} depicts the change points locations obtained by TF-PRISCA (top panel) and PRISCA (bottom panel) ($a_0=.001$, $L=30$, $\epsilon=10^{-5})$, its credible sets ($p=0.9$) along with PELT (point estimates, blue line). Rather than reporting the point estimates of PRISCA, we depict the credible sets coloring the time points. The segmentation of the two methods is practically identical, with two change points per year, one in autumn and one in spring. The methods differ in 2021, when PRISCA identifies an additional change point, with an extra segmentation in spring, while PELT estimates a single change point right in the middle of this segment. A visual inspection of the data suggests the possibility of an additional change point. The estimates of TF-PRISCA are identical to those of PELT. Our method reports confidence intervals as well, the average length of the credible sets is $10.7$ days.

		\section{Discussion}\label{sec:disc}
		
		In this article, we presented PRISCA, a computationally efficient Bayesian method to detect multiple changes in variances. As a point estimator, our proposal is as accurate, at times more accurate, than competitors. In addition, it provides credible sets to quantify estimation uncertainty that alternatives do not offer. While Bayesian methods naturally provide such a measure of uncertainty, PRISCA is orders of magnitude faster than existing Bayesian methods; it is more accurate and requires no MCMC. We also show that PRISCA has theoretical guarantees in terms of localization rate that competing Bayesian methods lack. 
		
		The paper contributes to making Bayesian change point methodologies more usable in applications. Practitioners do not commonly employ Bayesian change point methodologies because they are inherently slow, require extensive parameter tuning, and have no theoretical guarantees. Here, we address several of these issues for a specific type of distributional change, changes in the variance parameter of a Gaussian sequence model. \cite{wang2020simple} complement our proposal dealing with piecewise changes in the Gaussian mean. The two methods share the same rationale: solving the more complex multiple change point problem by recursively tackling the more accessible single change point case. The difference lies in that \cite{wang2020simple} sum together multiple single change point models, while we take a product of them. The reason is that they deal with a location parameter while we deal with a scale one. 	Nevertheless, the two works represent a template to generalize these methods beyond the Gaussian setting.
		
		The construction of PRISCA is driven by the need to develop a scalable algorithm for inference. To do that, several features that are part of standard methods have been dropped, \textit{e.g.} a minimum spacing condition. A naive modification of Algorithm~\ref{alg:prisca_fit} where one fits single change point models conditioning on the locations of the other estimates to enforce constraints will not work. To start, the algorithm would become heavily order-dependent -- early mistakes may have an impact on the final inference -- while our symmetric model together with backfitting may erase mistakes made early on. In addition, the characterization of PRISCA as a VB algorithm will cease to hold. An important direction of future work is how to incorporate these features while maintaining the good properties of the algorithm.
		
		In the paper, we proved that our point estimator is consistent  in the single change point problem. Such a result is the cornerstone to establish multiple change points results for frequentist estimators, with the technique of the proof depending on how the generalization to multiple change points is done; \textit{e.g.}, through binary segmentation or dynamic programming. This does not seem the case for Bayesian estimators, as, in the multiple change points problem, one essentially defines a probability distribution on the space of partitions. Marginal likelihoods cannot be computed in closed form, making it difficult to study the theoretical properties. This challenge can explain the relatively small literature on theoretical results for Bayesian change point methodologies.

\newpage
\begin{center}
	{\large\bf SUPPLEMENTARY MATERIAL}
\end{center}

\section{Auxiliary derivations of the single change point model in Section 2}  \label{sec:single_details}

We derive the posterior distribution parameters in Equations (2)--(4). We have to compute $\mathbb{P}(\bfgamma, \bftau| \bfy)=\mathbb{P}(\bftau| \bfy, \bfgamma) \mathbb{P}(\bfgamma| \bfy)$, where we omit the conditioning on parameters $\bfpi$, $a_0$ and $\sigma^2$ to simplify the notation. Conditionally on given change point location $\bfgamma$, say $\gamma_t=1$, inferring $s$ (recall $\bftau| \gamma_t=1$ is a vector with $1$s in the first $t-1$ positions, and all entries equal to $s$ in the remaining $T-t$) corresponds to  a standard Bayesian inference of the precision parameter with a Gamma prior; for background see \cite{gelman2013bayesian}. The posterior $\mathbb{P}(s^2|\gamma_t=1,\bfy) = \text{Gamma}(a_t,b_t)$, with parameters given in Equation (4) in the manuscript and below for completeness
\begin{equation*}\label{SMeq:posttau}
	a_t =a_0 + \frac{T-t+1}{2} \,\,\,\, \text{and} \,\,\,\, b_t=a_0+ \frac{\bfy_{t:T}^T \bfy_{t:T}}{2 \sigma^2}.
\end{equation*}
Now, $\bfgamma$ is a discrete random vector of length $T$ having one non-zero entry; \textit{i.e.},  $\mathbb{P}(\bfgamma| \bfy)$ is a Multinomial distribution with parameters $(1,\bfalpha)$, where $\bfalpha$ is a vector with entries $\alpha_t:=\mathbb{P}(\gamma_t=1|\bfy)$ and $\sum_{t=1}^T \alpha_t=1$. This can be computed as 
\begin{equation}\label{eq:bayesAlpha}
 \mathbb{P}(\gamma_t=1|\bfy) \propto \mathbb{P}(\bfy|\gamma_t=1) \mathbb{P}(\gamma_t=1) = \mathbb{P}(\bfy|\gamma_t=1) \pi_t, 
\end{equation}
where
 \[
\begin{array}{lll}
	\mathbb{P}\left(   \bfy |    \gamma_t = 1 \right) &=& \displaystyle   \int  	\mathbb{P}\left(   \bfy |    \gamma_t = 1,s \right) dP(s^2 |a_0 )\\
	& \propto& \displaystyle \int    \vert \Gamma \vert^{-1/2} \exp\left(   -\frac{1}{2}    \bfy^{\top } \Sigma^{-1} \bfy
	 \right) \text{Gamma}(s^2 |a_0,b_0) d s^2\\
	 
	 & \propto & \exp\bigg(   -\frac{1}{2\sigma^2} \sum_{j=1}^{t-1} y_j^2   \bigg)
	 \frac{\Gamma\big(  \frac{T-t+1}{2}+a_0 \big)}{\big(  b_0  +   \frac{1}{2\sigma^2} \sum_{j=t}^{T} y_j^2\big)^{ \frac{T-t+1}{2}+a_0  }} \nonumber
\end{array}
\]
where
\[
\Sigma \,=\, \left(  \begin{matrix}
	\sigma^2 I_{t-1} &0\\
	0&   \frac{\sigma^2}{s^2} I_{T-t+1}
\end{matrix}\right),
\]
where the proportionality constant above are also available in closed form but unnecessary when computing $\mathbb{P}(\gamma_t=1|\bfy)$ because they cancel with normalizing constant of \eqref{eq:bayesAlpha}, which is simply $\sum_{t=1}^T \mathbb{P}(\bfy|\gamma_t=1) \pi_t$. Putting everything together we get $(2)-(3)$ in the manuscript
\begin{align*}\label{eq:posterior}
	\bfgamma|\bfy &\sim  \text{Multinomial}(1,\bfalpha), \nonumber \\
		\alpha_t &= \frac{P(\bfy | \gamma_t=1,) \pi_t}{\sum_{\pi_j} P(\bfy | \gamma_j=1) \pi_j}.
\end{align*}

Given $\mathbb{P}(\bfgamma, \bftau| \bfy)$, we can compute posterior quantities. A crucial one for the multiple change point algorithm is $\overline{\bftau}^2:=E[\bftau^2|\bfy]= (E[\tau_1^2|\bfy],\ldots,E[\tau_T^2|\bfy])$ given in Equation $(7)$. Here
\[E[\tau^2_t|\bfy] = \int \tau^2_t  d P (\tau^2_t |\bfy)  =  \sum_{i=1}^t \left[ \int s^2 \mathbb{P} (s^2 |\bfy,\gamma_i=1) d s^2 \right] \mathbb{P}(\gamma_i=1|\bfy) +1 -\sum_{i=1}^{t}\mathbb{P}(\gamma_i=1|\bfy), \] 
where the above accounts for the fact that $\tau_t$ could be equal to $s$ or to $1$ depending on $\bfgamma$.

In the rest of the proof, we will simply consider the case $\pi_t = 1/T$ for $t=1,\ldots,T$. This means that $ \mathbb{P}(\gamma_t=1|\bfy) $ and $\mathbb{P}(\bfy|\gamma_t=1) $  are proportional to each other. Hence, in the rest of the proof we can simply work with the latter.  Finally, we emphasize that the general case $ c_{\min}/T \leq \pi_t \leq c_{\max}/T$ can be handled similarly, but to keep the notation simpler we focus on the case $\pi_t = 1/T$ for $t=1,\ldots,T$.


\section{Auxiliary lemmas for proof of  Theorem 1}
\begin{definition}
\label{def1}
A process $(X_t,t\in{\mathbb{Z}})$ is said to be $\alpha$-mixing if 
$$ \,\underset{k\rightarrow \infty}{\lim }\alpha_k \,    =\,
0,$$
where 
\[
\alpha_k( \{X_t\}_{t\in \mathbb{Z}}  )   \,:=\,  \, \sup_{t\in{\mathbb{Z}}}\alpha(\sigma(X_s,s\le t),\sigma(X_s,s\ge t+k))   .
\]
Here  the strong  mixing, or $\alpha$-mixing coefficient between two $\sigma$-fields $\mathcal{A}$ and $\mathcal{B}$ is defined as
$$
\alpha(\mathcal{A}, \mathcal{B})=\sup _{A \in \mathcal{A}, B \in \mathcal{B}}|\mathbb{P}(A \cap B)-\mathbb{P}(A) \mathbb{P}(B)| .
$$
We often simply write $\alpha_k$  to indicate $\alpha_k( \{X_t\}_{t\in \mathbb{Z}}  ) $.
\end{definition}


\begin{lemma}
	\label{lemma1}
		\textbf{[Lemma 3 in   \cite{madrid2024change} ] }
	Let $\nu>0$ be given. Suppose $\Big\{X_i\Big\}_{i=1}^{\infty}$ is a stationary $\alpha$-mixing time series with mixing coefficient $\alpha_k$ and that $\mathbb{E}(X_i)=0$.  	Suppose that there exists $\delta, \Delta>0$ such that
	$$
	\underset{i \in \mathbb{N}}{\sup}	\,\mathbb{E}\Big(\Big|X_i\Big|^{2+\delta+\Delta}\Big) \leq D_1
	$$
	and
	$$
	\sum_{k=0}^{\infty}(k+1)^{\delta / 2} \alpha_k^{\Delta /(2+\delta+\Delta)} \leq D_2 .
	$$
	Then for any $0<a<1$ it holds that
	$$
	\mathbb{P}\Big(\Big|\sum_{i=1}^r X_i\Big| \leq \frac{C}{a} \sqrt{r}\{\log (r \nu)+1\} \text { for all } r \geq 1 / \nu\Big) \geq 1-a^2,
	$$
	where $C$ is some absolute constant.
\end{lemma}

\begin{assumption} \label{ass1_s}
	Let $t_0$ be the time instance such that $Y_{t} \sim  F_1$ for $t \geq t_0$ and $Y_{t} \sim F_0$ for $t < t_0$, and let $\tau^2=\sigma^2_l/\sigma^2_r$ where  $\sigma^2_l:= \mathbb{E}(Y_1^2)$ and $\sigma^2_r :=  \mathbb{E}(Y_{t_0+1}^2) $.  We assume that $\{Y_t, t \in  \mathbb{N}\}$ is $\alpha $-mixing with coefficients satisfies $\alpha_k \,\leq \, e^{-C k }$ for some positive constant $C>0$. Also, we require that   $\max\{ \mathbb{E}(| Y_1^2-\sigma_l^2 |^{2+\delta  }),  \mathbb{E}(| Y_{t_0+1}^2- \sigma_r^2|^{2+\delta  }) \} \,\leq \,D_1$ for some positive constants $\delta,D_1$. In addition assume that: 
 	\begin{enumerate}[a.]
		\item There exists a constant $c>0$ such that $\min\{t_0,T-t_0\}>cT$.
		\item For some  fixed intervals $I_1 \subset (1,\infty)$  and $I_2 \subset (0,1)$   we have  that  $\tau^2 \in I_1\cup I_2$.  
		\item  The hyperparameters are $a_0>0$ and $\bfpi$ satisfies that  $\sum_t \pi_t=1$, and  $ c_{\min}/T \leq \pi_t \leq c_{\max}/T$ for all $t =1,\ldots,T$, and some positive constants $c_{\min}$ and $c_{\max}$. 
	\end{enumerate}
\end{assumption}

\begin{lemma}
		\label{lemma2}
	Suppose that  Assumption \ref{ass1_s} holds and  for $s< t_0$ let 
	
	\begin{equation}
		\label{eqn:error_terms_0}
		\begin{array}{lll}
			\tilde{\epsilon}_{t_0,s}    &:=&  \displaystyle \frac{1}{2\sigma_l^2}\sum_{j=s}^{t_0-1} Y_j^2   -    \frac{(t_0-s)}{2},\\
			\tilde{\epsilon}_{t_0}  &  :=&  \displaystyle   \frac{1}{2\sigma_l^2} \sum_{j=t_0}^{T} Y_j^2 -    \frac{1}{\tau^2} \frac{T-t_0+1}{2},\\
			\tilde{\epsilon}_s  & :=& \displaystyle  \frac{1}{2\sigma_l^2}\sum_{j=s}^{T} Y_j^2   - \frac{(t_0-s)}{2} -  \frac{1}{\tau^2} \frac{(T-t_0+1)}{2}.
		\end{array}
	\end{equation}
Then for any $a_T \in (0,1)$ for some constant $C>0$ we have that 
\begin{equation}
	\label{eqn:conclusion1}
	\mathbb{P}\left (   \, \vert \tilde{\epsilon}_{t_0,s}   \vert    \leq   \frac{ C \sqrt{ (t_0-s) }\log(t_0-s) }{a_T} ,\,\,\,\,\forall s<t_0 \right)  \,\geq \, 1- a_T^2,
\end{equation}

and 
\begin{equation}
	\label{eqn:conclusion2}
	\mathbb{P}\left (   \vert \tilde{\epsilon}_{t_0}   \vert    \leq   \frac{ C  \sqrt{ (T- t_0+1) }\log(T-t_0+1) }{c_0 \, a_T}  \right)  \,\geq \, 1- a_T^2,
\end{equation}
where $c_0$ is the smallest element of $I_2$. Furthermore, 
\begin{equation}
	\label{eqn:conclusion3}
	      \tilde{\epsilon}_s   =				\tilde{\epsilon}_{t_0,s}    + 		\tilde{\epsilon}_{t_0}.
\end{equation}
	
\end{lemma}

\begin{proof}
	First we prove (\ref{eqn:conclusion1}). Towards that end, notice that 
	\[
	   	\tilde{\epsilon}_{t_0,s}  \,=\,  \frac{1}{2}\left[     \sum_{j=s}^{t_0-1}   \left(     \frac{Y_j^2}{\sigma_l^2} -   \mathbb{E}\left(\frac{Y_j^2}{\sigma_l^2}\right)    \right)         \right].
	\]
	Thus $\epsilon_{t_0,s}$ is the sum of zero centered random variables.  Moreover,
	\[
	\sigma(\frac{Y^2_s}{\sigma_l^2}-1,s\le t)  \subset   \sigma(Y_s,s\le t) ,\,\,\,\,\,\,    \sigma(\frac{Y_s^2}{\sigma_l^2}-1,s\ge t+k) \,\subset \, \sigma(Y_s,s\ge t+k) 
	\]
	for all $t$ and $k$. Hence,
	\[
	\alpha_k( \{Y_t^2/\sigma_l^2-1 \}_{t\in \mathbb{Z}}  )     \,\leq\,  \alpha_k( \{Y_t \}_{t\in \mathbb{Z}}  )   \,\leq\, e^{-Ck}
	\]
	for all $k$.
	Therefore, an application of Lemma \ref{lemma1}  with $\nu =1$ leads to
	\[
	   	\mathbb{P}\left (  \, \vert \tilde{\epsilon}_{t_0s}   \vert    \leq   \frac{ \tilde{C} \sqrt{ (t_0-s) }\log(t_0-s) }{a_T} \,\,\,\,\forall s< t_0  \right)  \,\geq \, 1- a_T^2,
	\]
	where $a_T \in (0,1)$. This proves (\ref{eqn:conclusion1}).

	Next, we prove (\ref{eqn:conclusion2}). Notice that 
	\[
		\tilde{\epsilon}_{t_0}    \,=\, \frac{1}{2\tau^2} \left[ \sum_{j=t_0}^{T} \left( \frac{Y_j^2}{\sigma_r^2} -   \mathbb{E}\left( \frac{Y_j^2}{\sigma_r^2} \right)  \right)  \right].
	\]
	Hence, the same argument from above shows that
	\[
	\mathbb{P}\left(   \left \vert  \sum_{j=t_0}^{T} \left( \frac{Y_j^2}{\sigma_r^2} -   \mathbb{E}\left( \frac{Y_j^2}{\sigma_r^2} \right)  \right)  \right\vert  \leq  \frac{ \tilde{C} \sqrt{ (t_0-s) }\log(t_0-s) }{a_T} \right)\,\geq \, 1- a_T^2.
	\] 
	This in turn implies (\ref{eqn:conclusion2}).  
	
	Finally, (\ref{eqn:conclusion3})  follows by definition.
	 
\end{proof}

\begin{lemma}
	\label{lem4}
	Assume in addition to Assumption \ref{ass1_s} it holds that  $Y_{t} \iidsim N(0,\sigma_r^2)$ for $t \geq t_0$ and $Y_{t} \iidsim N(0,\sigma_l^2)$ for $t < t_0$. With the notation from Lemma \ref{lemma2} we have that 
		\[
	\mathbb{P}(\vert \tilde{\epsilon}_{t_0,s}  \vert  \geq    8\sqrt{(t_0-s)\log(t_0-s)} ) \,\leq\, \frac{2}{(t_0-s)^{8}}, \,\,\,\,\forall s<t_0,
	\]
	and 
	\[
		\mathbb{P}\left(   \vert  \tilde{\epsilon}_{t_0} \vert    \geq c_0^{-1} 8\sqrt{(T-t_0+1)\log(T-t_0+1)}\right)   \,\leq\, \frac{2}{(T-t_0+1)^8}.
	\]
\end{lemma}
\begin{proof}
	With the notation from (\ref{eqn:error_terms_0}), we have by the chi-squared concentration inequality, that 
	\begin{equation}
		\label{eqn:e5}
		\begin{array}{lll}
			\mathbb{P}(\vert \tilde{\epsilon}_{t_0,s}   \vert &\geq &   8\sqrt{(t_0-s)\log(t_0-s)} )  \,\leq \, 2\exp\left(  - (t_0-s)\left( 8\sqrt{\log(t_0-s)}/\sqrt{(t_0-s)} \right)^2 /8 \right) \\
			& =&\displaystyle  \frac{2}{(t_0-s)^{8}}.
		\end{array}
	\end{equation}
	Furthermore,   the chi-squared concentration inequality also implies that
	\begin{equation}%
		\label{eqn:e6}
	\begin{array}{lll}
		\mathbb{P}\left(   \vert  \tilde{\epsilon}_{t_0} \vert    \geq c_0^{-1} 8\sqrt{(T-t_0+1)\log(T-t_0+1)}\right)  & \leq &\mathbb{P}\left(   \vert  \tilde{\epsilon}_{t_0} \vert    \geq \tau^{-2}  8\sqrt{(T-t_0+1)\log(T-t_{0}+1)}\right)  \\
		& \leq & 2\exp\bigg(   -(T-t_0+1)\bigg(  8\sqrt{\log(T-t_0+1)} \\\ 
		& &/\sqrt{(T-t_0+1) }\bigg)^2/8  \bigg)\\
		& \leq&  \displaystyle \frac{2}{(T-t_0+1)^8}.
		\end{array}
	\end{equation}
The claim then follows.
\end{proof}

\begin{lemma}
	\label{lemma3}
	Suppose that  Assumption \ref{ass1_s} holds and  for $s> t_0$ let 
	
	\begin{equation}
		\label{eqn:error_terms_2}
		\begin{array}{lll}
			\tilde{\epsilon}_{t_0,s}    &:=&  \displaystyle \displaystyle \frac{1}{2\sigma_l^2}\sum_{j=t_0}^{s-1} Y_j^2   -    \frac{1}{\tau^2} \frac{(s-t_0)}{2},\\
			\tilde{\epsilon}_{s}  &  :=&  \displaystyle  \frac{1}{2\sigma_l^2}\sum_{j=s}^{T} Y_j^2   -  \frac{1}{\tau^2}  \frac{(T-s+1)}{2}, \\
			\tilde{\epsilon}_{t_0} & :=& \displaystyle  \frac{1}{2\sigma_l^2}\sum_{j=t_0}^{T} Y_j^2   -  \frac{1}{\tau^2} \frac{(T-t_0+1)}{2}.
		\end{array}
	\end{equation}

	Then for any $a_T \in (0,1)$ for some constant $C>0$ we have that 
	\begin{equation}
		\label{eqn:conclusion4}
		\mathbb{P}\left (   \vert \tilde{\epsilon}_{t_0,s}   \vert    \leq   \frac{ C \sqrt{ (s-t_0) }\log(s-t_0) }{ c_0 a_T},\,\,\,\,\,\forall s> t_0  \right)  \,\geq \, 1- a_T^2,
	\end{equation}
	
	and 
	\begin{equation}
		\label{eqn:conclusion5}
		\mathbb{P}\left (   \vert \tilde{\epsilon}_{s}   \vert    \leq   \frac{ C  \sqrt{ (T- s+1) }\log(T-s+1) }{c_0 \, a_T}  ,\,\,\,\,\,\,  \forall s>t_0 \right)  \,\geq \, 1- a_T^2,
	\end{equation}
	where $c_0$ is the smallest element of $I_2$. Furthermore, 
	\begin{equation}
		\label{eqn:conclusion6}
		\tilde{\epsilon}_{t_0}  =				\tilde{\epsilon}_{t_0,s}    + 		\tilde{\epsilon}_{s}.
	\end{equation}
	
\end{lemma}

\begin{proof}
The claim follows with a very similar argument to that in the proof of Lemma \ref{lemma2}.
\end{proof}

\begin{lemma}
	\label{lem5}
	Assume in addition to Assumption \ref{ass1_s} it holds that  $Y_{t} \iidsim N(0,\sigma_r^2)$ for $t \geq t_0$ and $Y_{t} \iidsim N(0,\sigma_l^2)$ for $t < t_0$. With the notation from Lemma \ref{lemma3} we have that 
	\[
	\mathbb{P}\left(   \vert  \tilde{\epsilon}_{s,t_0}\vert    \geq c_0^{-1} 8\sqrt{(s-t_0)\log(s-t_0)}\right) \,\leq\, \frac{2}{(s-t_0)^8},   \,\,\,
	\]
	and 
	\[
		\mathbb{P}\left(   \vert  \tilde{\epsilon}_{s}\vert    \geq c_0^{-1} 8\sqrt{(T-s+1)\log(T-s+1)}\right)\,\leq \, \frac{2}{(T-s+1)^8}
	\]
	for all $s> t_0$.
\end{lemma}

\begin{proof}
	Then, by the chi-squared concentration inequality
	
	\begin{equation}
		\label{eqn:e7}
		\begin{array}{lll}
			\mathbb{P}\left(   \vert  \tilde{\epsilon}_{s,t_0}\vert    \geq c_0^{-1} 8\sqrt{(s- t_0)\log(s-t_0)}\right)  & \leq &\mathbb{P}\left(   \vert  \tilde{\epsilon}_{s,t_0} \vert    \geq \tau^{-2}  8\sqrt{(s-t_0)\log(s-t_0)}\right)  \\
			& \leq & 2\exp\bigg(   -(s-t_0)\bigg(  8\sqrt{\log(s-t_0)} \\\ 
			& &/\sqrt{(s-t_0) }\bigg)^2/8  \bigg)\\
			& \leq&  \displaystyle \frac{2}{(s-t_0)^8}.
		\end{array}
	\end{equation}
	Similarly,

	\begin{equation}
		\label{eqn:e8}
		\begin{array}{lll}
			\mathbb{P}\left(   \vert  \tilde{\epsilon}_{s}\vert    \geq c_0^{-1} 8\sqrt{(T-s+1)\log(T-s+1)}\right)  & \leq &\mathbb{P}\left(   \vert  \tilde{\epsilon}_s \vert    \geq \tau^{-2}  8\sqrt{(T-s+1)\log(T-s+1)}\right)  \\
			& \leq & 2\exp\bigg(   -(T-s+1)\bigg(  8\sqrt{\log(T-s+1)} \\\ 
			& &/\sqrt{(T-s+1) }\bigg)^2/8  \bigg)\\
			& \leq&  \displaystyle \frac{2}{(T-s+1)^8}.
		\end{array}
	\end{equation}
	
\end{proof}

\section{Proof of Theorem 1}

We focus on the proof when the data are $\alpha$-mixing. This is based on Lemmas \ref{lemma2} and \ref{lemma3}. The case with Gaussian i.i.d data can be done with the same argument replacing Lemma \ref{lemma2} with Lemma \ref{lem4} and Lemma \ref{lemma3} with Lemma \ref{lem5}.

For simplicity of notation, we write $\sigma = \sigma_l$  and  $t = t_0$. Next notice that
\[
\begin{array}{lll}
	\mathbb{P}\left(   \bfY |    \gamma_t = 1, \bfpi, \sigma^2, a_0, b_0  \right) &=& \displaystyle   \int  	\mathbb{P}\left(   \bfY |    \gamma_t = 1, \bfpi, \sigma^2, a_0 ,b_0 , \tau\right) dP(\tau^2 |a_0 )\\
	& \propto& \displaystyle \int    \vert \Gamma \vert^{-1/2} \exp\left(   -\frac{1}{2}    \bfY^{\top } \Sigma^{-1} \bfY \right) \text{Gamma}(\tau^2 |a_0,b_0) d\tau^2\\
\end{array}
\]
where
\[
\Sigma \,=\, \left(  \begin{matrix}
	\sigma^2 I_{t-1} &0\\
	0&   \frac{\sigma^2}{\tau^2} I_{T-t+1}
\end{matrix}\right)
\]
with the notation $I_m$ indicating the $m\times m $ identity matrix. Notice that $\vert \Sigma\vert = (\sigma^2)^{t-1}  \left(  \sigma^2/\tau^2 \right)^{T-t+1}$, and hence
\[
\begin{array}{lll}
	\mathbb{P}\left(   \bfY |    \gamma_t = 1, \bfpi, \sigma^2, a_0 ,b_0 \right) &\propto& \displaystyle   \int  	\frac{ \tau^{\frac{T-t+1}{2}}   }{(\sigma^2)^{T/2}}\exp\left(   -\frac{1}{2\sigma^2} \sum_{j=1}^{t-1} Y_j^2   \right)\cdot\exp\bigg(   - \frac{\tau^2}{2\sigma^2}\sum_{j=t}^{T}Y_j^2  \bigg)\cdot  \\
	& &\displaystyle \frac{b_0^{a_0}   (\tau^2)^{a_0-1}  \exp(-b_0\tau^2) }{\Gamma(a_0)}d\tau^2\\
	&  = &\displaystyle  \frac{1}{\sigma^T}\exp\bigg(   -\frac{1}{2\sigma^2} \sum_{j=t}^{t-1} Y_j^2   \bigg) \frac{b_0^{ a_0 }}{\Gamma(a_0)} \int (\tau^2)^{ \frac{T-t+1}{2}   +a_0 -1 }\cdot\\
	& & \displaystyle \exp\bigg(-\tau^2 \bigg(b_0 + \frac{1}{2\sigma^2}\sum_{j=t}^T Y_j^2\bigg) \bigg) d\tau^2\\
	&  = &\displaystyle  \frac{1}{\sigma^T}\exp\bigg(   -\frac{1}{2\sigma^2} \sum_{j=1}^{t-1} Y_j^2   \bigg) \frac{b_0^{ a_0 }}{\Gamma(a_0)} 
	\frac{\Gamma\big(  \frac{T-t+1}{2}+a_0 \big)}{\big(  b_0  +   \frac{1}{2\sigma^2} \sum_{j=t}^{T} Y_j^2\big)^{ \frac{T-t+1}{2}+a_0  }}\\ 
	&\propto & \displaystyle \exp\bigg(   -\frac{1}{2\sigma^2} \sum_{j=1}^{t-1} Y_j^2   \bigg)
	\frac{\Gamma\big(  \frac{T-t+1}{2}+a_0 \big)}{\big(  b_0  +   \frac{1}{2\sigma^2} \sum_{j=t}^{T} Y_j^2\big)^{ \frac{T-t+1}{2}+a_0  }}\\ 
	& =: & \Delta_t.
\end{array}
\]


\subsection{Case $s < t$.}

Suppose that $t$ is the true change point and  $s<t$. Then
\[
\begin{array}{lll}
	\displaystyle    	 \frac{\Delta_t}{\Delta_s}  &  = &\displaystyle \exp\bigg(   -\frac{1}{2\sigma^2} \sum_{j=1}^{t-1} Y_j^2   \bigg)
	\frac{\Gamma\big(  \frac{T-t+1}{2}+a_0 \big)}{\big(  b_0  +   \frac{1}{2\sigma^2} \sum_{j=t}^{T} Y_j^2\big)^{ \frac{T-t+1}{2}+a_0  }}\\ 
	& &   	\displaystyle \cdot  \left[\exp\bigg(   -\frac{1}{2\sigma^2} \sum_{j=1}^{s-1} Y_j^2   \bigg)
	\frac{\Gamma\big(  \frac{T-s+1}{2}+a_0 \big)}{\big(  b_0  +   \frac{1}{2\sigma^2} \sum_{j=s}^{T} Y_j^2\big)^{ \frac{T-s+1}{2}+a_0  }}\right]^{-1}\\
	&\approx& \displaystyle   \frac{  \Gamma\big(  \frac{T-t+1}{2}\big)    }{  \Gamma\big(  \frac{T-t+1}{2} +  \frac{t-s}{2} \big) }\exp\big(    -\frac{1}{2\sigma^2} \sum_{j=s}^{t-1} Y_j^2  \big)\frac{  \big(  \frac{1}{2\sigma^2} \sum_{j=s}^{T} Y_j^2  \big)^{   \frac{T-s+1}{2}+a_0 }   }{ \big(  \frac{1}{2\sigma^2} \sum_{j=t}^{T} Y_j^2  \big)^{   \frac{T-t+1}{2}+a_0 }   }\\
	& \approx&  \displaystyle \frac{1}{\big(   \frac{T-t+1}{2} \big)^{  \frac{t-s}{2}  } }\exp\big(    -\frac{1}{2\sigma^2} \sum_{j=s}^{t-1} Y_j^2  \big)\frac{  \big(  \frac{1}{2\sigma^2} \sum_{j=s}^{T} Y_j^2  \big)^{   \frac{T-s+1}{2}+a_0 }   }{ \big(  \frac{1}{2\sigma^2} \sum_{j=t}^{T} Y_j^2  \big)^{   \frac{T-t+1}{2}+a_0 }   }\\
\end{array}
\]
where the third equality follows by the properties of the gamma function.  Next, let
\begin{equation}
	\label{eqn:error_terms}
	\begin{array}{lll}
		\tilde{\epsilon}_{t,s}    &:=&  \displaystyle \frac{1}{2\sigma^2}\sum_{j=s}^{t-1} Y_j^2   -    \frac{(t-s)}{2},\\
		\tilde{\epsilon}_t   &  :=&  \displaystyle   \frac{1}{2\sigma^2} \sum_{j=t}^{T} Y_j^2 -    \frac{1}{\tau^2} \frac{T-t+1}{2},\\
		\tilde{\epsilon}_s  & :=& \displaystyle  \frac{1}{2\sigma^2}\sum_{j=s}^{T} Y_j^2   - \frac{(t-s)}{2} -  \frac{1}{\tau^2} \frac{(T-t+1)}{2},
	\end{array}
\end{equation}

\[
\begin{array}{lll}
	\displaystyle 	\log\left(  \frac{\Delta_t}{\Delta_s} \right) &\approx&\displaystyle  -\frac{(t-s)}{2}\log\left(  \frac{T-t+1}{2} \right)- \frac{(t-s)}{2}
	-\tilde{\epsilon}_{t,s} +    \\
	& &  \displaystyle \left( \frac{T-s+1}{2}+a_0\right)\log\left(   \frac{(t-s)}{2} + \frac{1}{\tau^2}\frac{(T-t+1)}{2}  + \tilde{\epsilon}_s\right)\\
	& \displaystyle &-\left(\frac{T-t+1}{2}+a_0\right)\log\left( \frac{1}{\tau^2}\frac{(T-t+1)}{2}+ \tilde{\epsilon}_t\right).\\
\end{array}
\]
Next we set $v =1/\tau^2$, and write

\[
\begin{array}{lll}
	G(v )&\,:=\,& \displaystyle -\frac{(t-s)}{2}\log\left(  \frac{T-t+1}{2} \right)- \frac{(t-s)}{2}
	-\tilde{\epsilon}_{t,s} +    \\
	&&	\displaystyle \left( \frac{T-s+1}{2}+a_0\right)\log\left(   \frac{(t-s)}{2} + v\frac{(T-t+1)}{2}  + \tilde{\epsilon}_s\right)\\
	& &\displaystyle -\left(\frac{T-t+1}{2}+a_0\right)\log\left(v\frac{(T-t+1)}{2}+ \tilde{\epsilon}_t\right).\\
\end{array}
\]
%

For the constant $C>0$ in Lemma \ref{lemma2}, consider the event $\Omega$ given as
\[
\begin{array}{lll}
	\Omega &\,:=\,& \left\{    \vert \tilde{\epsilon}_{t,s}   \vert    \,\leq \,  \frac{ C \sqrt{ (t-s) }\log(t-s) }{a_T} ,\,\,\,\forall   \,\,   s< t,  \,  s \geq cT,   \,\,  t-s \geq  \sqrt{T }\log^{1+\epsilon} T \right\} \cap\\
	& & \left\{    \vert \tilde{\epsilon}_{t}   \vert    \,\leq \,\frac{ C  \sqrt{ (T- t+1) }\log(T-t+1) }{c_0 \, a_T}    \,\,\right\}. 
\end{array}
\]
Then Lemma  \ref{lemma2} implies that 
\[
\mathbb{P}\left(   \Omega  \right)\,\geq \,   1  - 2 a_T^2.
\]
Suppose now that $\Omega$ holds. Then   notice that $G(v) \geq G_l(v)$ where
\[
\begin{array}{lll}
	G_l(v )&\,:=\,& \displaystyle -\frac{(t-s)}{2}\log\left(  \frac{T-t+1}{2} \right)- \frac{(t-s)}{2}
	-\tilde{\epsilon}_{t,s} +    \\
	&&	\displaystyle \left( \frac{T-s+1}{2}+a_0\right)\log\left(   \frac{(t-s)}{2} + v\frac{(T-t+1)}{2}  + \tilde{\epsilon}_{t,s} -    C (c_0 a_T)^{-1}\sqrt{T }\log T\right)\\
	& &\displaystyle -\left(\frac{T-t+1}{2}+a_0\right)\log\left(v\frac{(T-t+1)}{2}+    C (c_0 a_T)^{-1}\sqrt{T }\log T\right).\\
\end{array}
\]
We proceed to  show that $G_l(v)>0$ when $v$ is sufficiently far away from $1$. The first derivative of $G_l$ evaluated at an arbitrary point $v$ is  
\begin{equation}
	\label{eqn:lower}
	\begin{array}{lll}
		G_l^{\prime}(v)  & = &  \displaystyle \left(   \frac{T-s+1}{2}+a_0 \right)\cdot \frac{  \frac{T-t+1}{2}   }{  \frac{(t-s)}{2}  +   \frac{v(T-t+1)}{2} +  \tilde{\epsilon}_{t,s}-  C (c_0 a_T)^{-1}\sqrt{T }\log T  }-\\
		& & \displaystyle	\left(   \frac{T-t+1}{2}+a_0 \right)\cdot \frac{  \frac{T-t+1}{2}   }{     \frac{v(T-t+1)}{2} +   C (c_0 a_T)^{-1}\sqrt{T }\log T },\\
		&=& \displaystyle	\frac{  \frac{T-t+1}{2}   }{D_l(v)}\Big[ (v-1)\frac{T-t+1}{2}\, \frac{t-s}{2} +  C (c_0 a_T)^{-1}\sqrt{T }\log T \frac{(2T -t-s+2)}{2} - \tilde{\epsilon}_{t,s} \left(\frac{T-t+1}{2}+a_0 \right)\Big],  \\
	\end{array}
\end{equation}
where $D_l(v)$ is the common denominator. It is easy to see that $D_l(v)>0$ for all $v  =  1/\tau^2 \in    I_1 \cup I_2$. The first two terms  in the square bracket dominate  the other because  $|\tilde{\epsilon}_{t,s}| \leq C a_T^{-1} \sqrt{(t-s)}\log(t-s)$. We get that $G'_l(v)>0$ for $\tau^2\in I_2$ and  $G'_l(v)<0$ for $\tau^2 \in I_1$. The derivative behaves like $(t-s)$ because $D_l(v)$ is  of order $T^2$.  Hwere, we have used the fact that we can take $a_T \rightarrow 0$ such that $\log^{\xi} T a_T \,\rightarrow \,\infty$.

Now we lower bound $G_l(v)$, 
\[
\begin{array}{lll}
	G_l(v)
	&=& \displaystyle  -\frac{(t-s)}{2}
	-\epsilon_{t,s} + \frac{(t-s)}{2} \log \left(\frac{ \frac{(t-s)}{2} - (1-v) \frac{(T-t+1)}{2}+  \tilde{\epsilon}_{t,s}-  C (a_Tc_0)^{-1}\sqrt{T} \log T }{\frac{(T-t+1)}{2}}+1 \right)+\\
	&& \displaystyle\left(\frac{T-t+1}{2}+a_0\right) \Bigg[\log \left( \frac{(t-s)}{2} + v\frac{(T-t+1)}{2}+  \tilde{\epsilon}_{t,s}-  C (a_Tc_0)^{-1}\sqrt{T} \log T\right) - \\
	&& \displaystyle  \log \left(v\frac{(T-t+1)}{2}+  C (a_Tc_0)^{-1}\sqrt{T }\log T \right) \Bigg]\\
	& \gtrsim& \displaystyle \Bigg[\left(-\frac{(t-s)}{2}
	-\tilde{\epsilon}_{t,s} \right) \left(\frac{(t-s)}{2} + v\frac{(T-t+1)}{2}  + \tilde{\epsilon}_{t,s}-  C (a_Tc_0)^{-1}\sqrt{T }\log T\right ) +\\
	&& \displaystyle \frac{(t-s)}{2} \left( \frac{(t-s)}{2} - (1-v) \frac{(T-t+1)}{2}+  \tilde{\epsilon}_{t,s}-  C (a_Tc_0)^{-1}\sqrt{T }\log T\right)+\\
	&& \displaystyle
	\left(\frac{T-t+1}{2}+a_0\right)
	\left(\frac{(t-s)}{2} +  \tilde{\epsilon}_{t,s} -    2 C (a_Tc_0)^{-1}\sqrt{T }\log T \right)\Bigg]\frac{1}{T}\\	
\end{array}
\]
\[
\begin{array}{lll}
	& \geq & \displaystyle  \Bigg[-\left(\frac{t-s}{2}+\tilde{\epsilon}_{t,s}+(v-1)\frac{T-t+1}{2}-  C (a_Tc_0)^{-1}\sqrt{T }\log T\right) \tilde{\epsilon}_{t,s}
	\\	
	&&\displaystyle +\left(\frac{t-s}{2}+\tilde{\epsilon}_{t,s}- 2C (a_Tc_0)^{-1}\sqrt{T }\log T\right)a_0-
	2C (a_Tc_0)^{-1}\sqrt{T} \log T \frac{T-t+1}{2} \Bigg]\frac{1}{T},
\end{array}
\]

where in first inequality we lower bound the middle term through the inequality $\log(1+x)\geq x/(x+1)$ for $x> -1$, the last term using the mean value theorem, and the common denominator by $T$.  We consider two Taylor's  expansions. In the first one, we apply Taylor's theorem to write $G_l(v)$ doing the expansion at a point $v_l$ such that $v_l=1-\delta_1$ and $\delta_1$ is positive and small and $v <v_l<1$.  Thus, we write
$$G_l(v)=G_l(v_l)+G'_l(v')(v-v_l),$$
with $v < v' < v_l. $  Then from (\ref{eqn:lower}),  we have that $G'_l(v')(v-v_l)$ is positive and behaves like $(t-s)$ for all $v$ not too close to $v_l$. Thus,   $G'_l(v')(v-v_l)  \asymp (t-s)$ for all $0 < v < v_l -\delta_2$  for some $\delta_2>0$. Since $G_l(v_l)\geq - \sqrt{T }\log T$,  the condition $(t-s) \geq \sqrt{T}\log^{1+\epsilon}T$  implies that $G_l(v)>0$ for all $0 < v < v_l -\delta_2$,   and we have the desired rate of convergence.

The argument on the opposite side ($v>1$) is symmetrical. We can do the Taylor expasion at a point $v_u> 1+\delta_1$ for an appropriate $\delta_1>0$, and have $G_l(v)>0$ for all $v>v_u+\delta_2$ for some $\delta_2 >0$. 

To sum up,  there exists $\delta>0 $, such  that $G(v)>0$ for all $0<v<1-2\delta$ and $v>1+2\delta$.

\subsection{Case $s>t$}

We now proceed to show that if $t$ is the location of the true change point and  $\mid t-s\mid   \gtrsim \sqrt{T }\log^{1+\epsilon} T$ then  $\Delta_t  > \Delta_s$. Towards that end,  as before
\[
\begin{array}{lll}
	\displaystyle    \frac{\Delta_t}{\Delta_s}  & = & \displaystyle   \exp\left(     \frac{1}{2\sigma^2}   \sum_{j=t}^{s-1} y_j^2   \right)\frac{\Gamma \left( \frac{T-t+1}{2} +a_0   \right) }{\Gamma \left( \frac{T-s+1}{2} +a_0   \right)} \frac{\left(        b_0 +     \frac{1}{2\sigma^2}    \sum_{j=s}^{T}    Y_j^2  \right)^{  \frac{T-s+1}{2}  +a_0  }  }{\left(        b_0 +     \frac{1}{2\sigma^2}    \sum_{j=t}^{T}    Y_j^2  \right)^{  \frac{T-t+1}{2}  +a_0  }  }
\end{array}
\]
and so
\[
\begin{array}{lll}
	\displaystyle  \log\left(  \frac{\Delta_t}{\Delta_s}\right)  & = & \displaystyle    \frac{1}{2\sigma^2}   \sum_{j=t}^{s-1} y_j^2   +  \log \left( \frac{\Gamma \left( \frac{T-t+1}{2} +a_0   \right) }{\Gamma \left( \frac{T-s+1}{2} +a_0   \right)}\right)\\
	& &\displaystyle 	\left(\frac{T-s+1}{2}  +a_0\right) \log \left( b_0 +     \frac{1}{2\sigma^2}    \sum_{j=s}^{T}    Y_j^2    \right)   \\
	& &\displaystyle \,-\,  \left(\frac{T-t+1}{2}  +a_0\right)  \log \left( b_0 +     \frac{1}{2\sigma^2}    \sum_{j=t}^{T}    Y_j^2  \right). \\
\end{array}
\]
Next for $v = 1/\tau^2$ define 
$$	\tilde{\epsilon}_{s,t}    :=  \displaystyle \frac{1}{2\sigma^2}\sum_{j=t}^{s-1} Y_j^2   -  v \frac{(s-t)}{2},$$
\[
\tilde{\epsilon}_s  \,:=\,      \displaystyle \frac{1}{2\sigma^2}\sum_{j=s}^{T} Y_j^2   -  v \frac{(T-s+1)}{2},
\]
and
\[
\tilde{\epsilon}_t \,:=\,   \displaystyle \frac{1}{2\sigma^2}\sum_{j=t}^{T} Y_j^2   -  v \frac{(T-t+1)}{2}.
\]



Then we define the event $\mathcal{E}$ as 
\[
\begin{array}{lll}
	\mathcal{E}&\,:=\,& \left\{    \vert \tilde{\epsilon}_{s,t}   \vert    \,\leq \,  C(c_0 a_T)^{-1}\sqrt{(s-t)}\log(s-t) ,\,\,\,\forall   \,\,   s>t,  \,  s \leq cT,   \,\,  t-s \geq  \sqrt{T }\log^{1+\epsilon} T\right\} \cap\\
	& & \left\{    \vert \tilde{\epsilon}_{s}   \vert    \,\leq \,  C(c_0 a_T)^{-1}\sqrt{(T-t+1)}\log(T-t+1), \,\,\forall   \,\,   s>t,  \,  s \leq cT\right\} 
\end{array}
\]
for some $c \in (0,1)$, and by Lemma \ref{lemma3} we have that 
\[
\mathbb{P}\left(   \mathcal{E}  \right)\,\geq \,   1 -  2a_T.
\]
Next, let us assume that $\mathcal{E}$ holds. Then

\[
\begin{array}{lll}
	\displaystyle  \log\left(  \frac{\Delta_t}{\Delta_s}\right)  & \approx &  \displaystyle    \frac{1}{2\sigma^2}   \sum_{j=t}^{s-1} y_j^2   +  \log \left(\left( \frac{T-s+1}{2} \right)^{(s-t)/2} \right)\\
	& &\displaystyle 	\left(\frac{T-s+1}{2} \right) \log \left(    \frac{1}{2\sigma^2}    \sum_{j=s}^{T}    Y_j^2    \right)   \\
	& &\displaystyle \,-\,  \left(\frac{T-t+1}{2}\right)  \log \left(      \frac{1}{2\sigma^2}    \sum_{j=t}^{T}    Y_j^2  \right)\\
\end{array}
\]
\[
\begin{array}{lll}
	& \geq& \displaystyle 	 \tilde{\epsilon}_{s,t}     + v \frac{(s-t)}{2}\,+\, \left(  \frac{s-t}{2}\right)\log \left(   \frac{T-s+1}{2} \right) \\
	& &\displaystyle  -  \Big(\frac{T-t+1}{2}+a_0 \Big)\log  \Big(   \tilde{\epsilon}_{t} + v\frac{T-t+1}{2}\Big) \\
	&&\displaystyle    	  +\Big(\frac{T-s+1}{2} +a_0 \Big) \log \Bigg(    \tilde{\epsilon}_{s} + v\frac{T-s+1}{2} \Bigg)\\
	& =:& \displaystyle   H(v).
\end{array}
\]
We follow the same step as before. We define a lower bound 
\begin{align} \label{eq:target}
	H_l(v)
	:=&  -  C (c_0 a_T)^{-1}\sqrt{(s-t)} \log (s-t)  + v\frac{(s-t)}{2} + \Big(\frac{s-t}{2}\Big) \log 
	\Big(  \frac{T-s+1}{2}\Big) \\
	&-  \Big(\frac{T-t+1}{2} \Big)\log  \Big(   - C (c_0 a_T)^{-1}\sqrt{T} \log T + v\frac{T-t+1}{2}\Big) +\\
	& \displaystyle
	\Big(\frac{T-s+1}{2} \Big) \log \Bigg( - C (c_0 a_T)^{-1}\sqrt{T}\log T + v\frac{T-s+1}{2} \Bigg),
	\nonumber
\end{align}
and notice that

\[
\begin{array}{lll}
	H_l(v)  & =  & \displaystyle -C (c_0 a_T)^{-1}\sqrt{(s-t)} \log (s-t)  + v\frac{(s-t)}{2} + \\
	& &\displaystyle \Big(\frac{T-s+1}{2}\Big) \log \left(\frac{v\frac{T-s+1}{2} -C (c_0 a_T)^{-1}\sqrt{T} \log T }{v\frac{T-t+1}{2} -C (c_0 a_T)^{-1}\sqrt{T }\log T }\right) +
	\\
	&   & \displaystyle	\Big(\frac{s-t}{2}\Big) \log \left(\frac{\frac{T-s+1}{2}  }{v\frac{T-t+1}{2} -C (c_0 a_T)^{-1}\sqrt{T }\log T   }\right) \\
	& =   & \displaystyle  - C (c_0 a_T)^{-1}\sqrt{(s-t) }\log (s-t)  + v\frac{(s-t)}{2} +   \\
	&&\displaystyle \Big(\frac{T-s+1}{2}\Big) \log\left( \frac{ v \frac{(t-s)}{2}   }{v\frac{T-t+1}{2} -C (c_0 a_T)^{-1}\sqrt{T }\log T     }   + 1\right) + \\
	&   & \displaystyle	\Big(\frac{s-t}{2}\Big) \log \left(\frac{\frac{T-s+1}{2} }{v\frac{T-t+1}{2} -C (c_0 a_T)^{-1}\sqrt{T }\log T}\right) \\
	& \geq& \displaystyle  - C (c_0 a_T)^{-1}\sqrt{(s-t) }\log (s-t)  + v\frac{(s-t)}{2} +   \\
	& &\displaystyle \Big(\frac{T-s+1}{2}\Big) \frac{ \frac{(t-s)}{2}   }{v\frac{T-s+1}{2} -C (c_0 a_T)^{-1}\sqrt{T} \log T}   +\\
	&& \displaystyle  \left(\frac{s-t}{2}\right)\left(  \frac{   \frac{T-s+1}{2}  }{\frac{T-s+1}{2}    +   v\frac{T-t+1}{2} -C (c_0 a_T)^{-1}\sqrt{T }\log T} \right)\\
	& \geq&\displaystyle  \frac{1}{T^2}\bigg\{  \bigg[ - C (c_0 a_T)^{-1}\sqrt{(s-t) }\log (s-t) + v\frac{(s-t)}{2}   \bigg] \cdot\\\ 
	& &  \,\,\,\,\,\,\,\,\,\,\,\,\bigg[ \frac{T-s+1}{2}    +   v\frac{T-t+1}{2} -C (c_0 a_T)^{-1}\sqrt{T }\log T\bigg]\cdot\\ 
	& &     \,\,\,\,\,\,\,\,\,\,\,\,\bigg[ v\frac{T-s+1}{2} -C (c_0 a_T)^{-1}\sqrt{T }\log T\bigg]\,+\, \bigg[\frac{T-s+1}{2}\bigg] \cdot \bigg[\frac{t-s}{2}\bigg]\cdot\\
	 & &  \,\,\,\,\,\,\,\,\,\,\,\,\bigg[ \frac{T-s+1}{2}    +   v\frac{T-t+1}{2} -C (c_0 a_T)^{-1}\sqrt{T }\log T\bigg]\\
	& &  \,\,\,\,\,\,\,\,\,\,\,\,    +    \bigg[  \frac{s-t}{2}\bigg]\cdot \bigg[\frac{T-s+1}{2} \bigg]  \cdot\bigg[v\frac{T-s+1}{2} -C (c_0 a_T)^{-1}\sqrt{T }\log T\bigg]  \bigg\}\\ 
	&\gtrsim & \displaystyle - \sqrt{T }\log T.
\end{array}
\]
Now, we look at the first derivative of  $H_l(v)$. We notice that 
\[
\begin{array}{lll}
	H_l^{\prime}(v) &  =&\displaystyle  \frac{(s-t)}{2} +  \Big(\frac{T-s+1}{2} \Big)\left(\frac{\frac{T-s+1}{2}}{ v\frac{T-s+1}{2}  - C (c_0 a_T)^{-1}\sqrt{T }\log T} \right) - \\
	&  & \displaystyle \Big(\frac{T-t+1}{2}  \Big)\left(\frac{\frac{T-t+1}{2}}{ v\frac{T-t+1}{2} -C (c_0 a_T)^{-1}\sqrt{T} \log T} \right)\\
	& =&\displaystyle  \frac{1}{D(v)}\bigg[  \frac{(s-t)}{2}\left( v\frac{T-s+1}{2} -C (c_0 a_T)^{-1}\sqrt{T }\log T\right) \left( v\frac{T-t+1}{2} -C (c_0 a_T)^{-1}\sqrt{T} \log T\right)+\\   
	& &\displaystyle \,\,\,\,\,\,\,\,\,\,\,\,\,\,\,   \left( \frac{T-s+1}{2}\right)^2\left[ v\frac{T-t+1}{2} -C (c_0 a_T)^{-1}\sqrt{T }\log T\right]\,+\,\\
	& &\displaystyle  \,\,\,\,\,\,\,\,\,\,\,\,\,\,\,   -\left( \frac{T-t+1}{2}\right)^2\left[ v\frac{T-s+1}{2} -C (c_0 a_T)^{-1}\sqrt{T }\log T\right] \bigg]\\
	& = & \displaystyle \frac{1}{D (v)} \Bigg[ v^2 \frac{(s-t)}{2}\frac{(T-s+1)}{2}\frac{(T-t+1)}{2}- \Big(\frac{T-t+1}{2} \Big)^2  v\frac{(T-s+1)}{2} +\\
	& & \displaystyle \,\,\,\,\,\,\,\,\,\,\,\,\,\,\, \Big(\frac{T-s+1}{2} \Big)^2 v\frac{(T-t+1)}{2} +\tilde{\xi}_{s,t}	\Bigg]\\
	& = & \displaystyle \frac{1}{D (v)} \Bigg[ v(v-1)\frac{(s-t)}{2} \frac{(T-s+1)}{2}\frac{(T-t+1)}{2} +\tilde{\xi}_{s,t}	 \Bigg]\\
\end{array}
\]
with $\tilde{\xi}_{s,t} = O(  T^2(a_T)^{-1}\sqrt{T}\log T)$, and $D(v)  \asymp T^2$. Hence, $H^{\prime}_l(v ) > 0$ if $v>1$ and $H^{\prime}_l(v ) <0$ if $v<1$. The proof concludes as in the previous case. 

\bigskip 

\textbf{Gaussian case.} In the case in which each observation $y_t$ follows a normal distribution, we follow the same proof but replacing Lemma  \ref{lemma2} with Lemma \ref{lem4} and Lemma \ref{lemma3} with Lemma \ref{lem5}.

\subsection{Proof of Corollary 1}

 Corollary 1 follows directly from the proof of Theorem 1, replacing Lemma \ref{lemma2} and Lemma \ref{lemma3} with  Lemma \ref{lemma6}  and  Lemma \ref{lemma7}, respectively. The latter two lemmas are given below.

\begin{lemma}
	\label{lemma6}
	For $s< t_0$ let 
	
	\begin{equation}
		\label{eqn:error_terms_02}
		\begin{array}{lll}
			\tilde{\epsilon}_{t_0,s}    &:=&  \displaystyle \frac{1}{2\sigma_l^2}\sum_{j=s}^{t_0-1} \hat{Y}_j^2   -    \frac{(t_0-s)}{2},\\
			\tilde{\epsilon}_{t_0}  &  :=&  \displaystyle   \frac{1}{2\sigma_l^2} \sum_{j=t_0}^{T} \hat{Y}_j^2 -    \frac{1}{\tau^2} \frac{T-t_0+1}{2},\\
			\tilde{\epsilon}_s  & :=& \displaystyle  \frac{1}{2\sigma_l^2}\sum_{j=s}^{T} \hat{Y}_j^2   - \frac{(t_0-s)}{2} -  \frac{1}{\tau^2} \frac{(T-t_0+1)}{2}.
		\end{array}
	\end{equation}
	Then for any $a_T \rightarrow 0$ for some constant $C>0$ we have that 
	\begin{equation}
		\label{eqn:conclusion12}
		\mathbb{P}\left (   \, \vert \tilde{\epsilon}_{t_0,s}   \vert    \leq   \frac{ C \sqrt{ T}\log T }{a_T} ,\,\,\,\,\forall s<t_0 \right)  \,\rightarrow \, 1,
	\end{equation}
	
	and 
	\begin{equation}
		\label{eqn:conclusion22}
		\mathbb{P}\left (   \vert \tilde{\epsilon}_{t_0}   \vert    \leq   \frac{ C  \sqrt{ (T- t_0+1) }\log(T-t_0+1) }{c_0 \, a_T}  \right)  \,\rightarrow \, 1,
	\end{equation}
	where $c_0$ is the smallest element of $I_2$. Furthermore, 
	\begin{equation}
		\label{eqn:conclusion32}
		\tilde{\epsilon}_s   =				\tilde{\epsilon}_{t_0,s}    + 		\tilde{\epsilon}_{t_0}.
	\end{equation}
	
\end{lemma}

\begin{proof}
	Notice that 
	\[
	\begin{array}{lll}
		\displaystyle \frac{1}{2\sigma_l^2}\sum_{j=s}^{t_0-1} \hat{Y}_j^2   -    \frac{(t_0-s)}{2}    & = & \displaystyle \frac{1}{2\sigma_l^2}\sum_{j=s}^{t_0-1} ( Y_j + \mu_j   - \hat{\mu}_j )^2   -    \frac{(t_0-s)}{2}\\
		 & = &\displaystyle \frac{1}{2\sigma_l^2}\sum_{j=s}^{t_0-1}  Y_j^2   -    \frac{(t_0-s)}{2}    \,+\,   \frac{1}{2\sigma_l^2}\sum_{j=s}^{t_0}  (\mu_j   - \hat{\mu}_j)^2    \,+\,   \frac{1}{\sigma_l^2}\sum_{j=s}^{t_0} Y_j  (\mu_j   - \hat{\mu}_j)\\
		  & = : & \displaystyle \frac{1}{2\sigma_l^2}\sum_{j=s}^{t_0-1}  Y_j^2   -    \frac{(t_0-s)}{2}   \, +\, \text{Err}_{s,t_0}\,\\
	\end{array}
	\]
	where 
	\[
    \begin{array}{lll}
    		\vert \text{Err}_{s,t_0}\vert    & \leq &  \displaystyle    \frac{1}{2\sigma_l^2}\sum_{j=s}^{t_0}  (\mu_j   - \hat{\mu}_j)^2     \,+\,    \frac{1}{\sigma_l^2}\left\vert    \sum_{j=s}^{t_0} Y_j  (\mu_j   - \hat{\mu}_j)\right\vert \\ 
    		  & \leq& \displaystyle  \frac{1}{2\sigma_l^2}\sum_{j=1}^{T}  (\mu_j   - \hat{\mu}_j)^2     \,+\,    \frac{1}{\sigma_l^2}\left\vert    \sum_{j=s}^{t_0} Y_j  (\mu_j   - \hat{\mu}_j)\right\vert  
    \end{array}
	\]
	Then, by Lemma \ref{lemma1} and our assumption on $\{\hat{\mu}_t\}$ we obtain that 
	\[
	 \underset{s<t_0}{\max}\,	\vert \text{Err}_{s,t_0}\vert    \,\leq\, c \frac{\sqrt{T} \log T}{a_T}
	\]
	for some positive constant $c>0$ with high probability.  The claim in (\ref{eqn:conclusion12}) then follows directly from Lemma \ref{lemma2}.
	
	To verify (\ref{eqn:conclusion22}), we observe that
	\[
	   \begin{array}{lll}
	   \displaystyle   \frac{1}{2\sigma_l^2} \sum_{j=t_0}^{T} \hat{Y}_j^2 -    \frac{1}{\tau^2} \frac{T-t_0+1}{2}  &  = &	   \displaystyle   \frac{1}{2\sigma_l^2} \sum_{j=t_0}^{T} (Y_j +  \mu_j - \hat{\mu}_j )^2 -    \frac{1}{\tau^2} \frac{T-t_0+1}{2} \\
	    & = &\displaystyle   \frac{1}{2\sigma_l^2} \sum_{j=t_0}^{T} Y_j^2 -    \frac{1}{\tau^2} \frac{T-t_0+1}{2}   \,+\, \frac{1}{2\sigma_l^2} \sum_{j=t_0}^{T} (\mu_j - \hat{\mu}_j )^2  \,+\, \frac{1}{2\sigma_l^2} \sum_{j=t_0}^{T} Y_j(\mu_j - \hat{\mu}_j )\\
	     & =:&\displaystyle   \frac{1}{2\sigma_l^2} \sum_{j=t_0}^{T} Y_j^2 -    \frac{1}{\tau^2} \frac{T-t_0+1}{2}   \,+\, \text{Err}_{t_0}\\
	   \end{array}
	\]
	with 
	\[
	 \begin{array}{lll}
	 	\vert   \text{Err}_{t_0}\vert   & \leq& \displaystyle    \frac{1}{2\sigma_l^2}\sum_{j=t_0}^{T}  (\mu_j   - \hat{\mu}_j)^2     \,+\,    \frac{1}{\sigma_l^2}\left\vert    \sum_{j=t_0}^{T} Y_j  (\mu_j   - \hat{\mu}_j)\right\vert \\ 
	 	& \leq& \displaystyle  \frac{1}{2\sigma_l^2}\sum_{j=1}^{T}  (\mu_j   - \hat{\mu}_j)^2     \,+\,       \frac{1}{\sigma_l^2} \left\vert    \sum_{j=t_0}^{T} Y_j  (\mu_j   - \hat{\mu}_j)\right\vert . 
	 \end{array}
	\]
	Hence, by  Lemma \ref{lemma1} and our assumption on $\{\hat{\mu}_t\}$ we obtain that 
	\[
		\vert \text{Err}_{t_0}\vert    \,\leq\, \frac{c \sqrt{T} \log T}{a_T}
	\]
	for some positive constant $c>0$ with high probability.  The claim follows. 
\end{proof}

\begin{lemma}
	\label{lemma7}
	Suppose that  Assumption \ref{ass1_s} holds and  for $s> t_0$ let 
	
	\begin{equation}
		\label{eqn:error_terms_22}
		\begin{array}{lll}
			\tilde{\epsilon}_{t_0,s}    &:=&  \displaystyle \displaystyle \frac{1}{2\sigma_l^2}\sum_{j=t_0}^{s-1} \hat{Y}_j^2   -    \frac{1}{\tau^2} \frac{(s-t_0)}{2},\\
			\tilde{\epsilon}_{s}  &  :=&  \displaystyle  \frac{1}{2\sigma_l^2}\sum_{j=s}^{T} \hat{Y}_j^2   -  \frac{1}{\tau^2}  \frac{(T-s+1)}{2}, \\
			\tilde{\epsilon}_{t_0} & :=& \displaystyle  \frac{1}{2\sigma_l^2}\sum_{j=t_0}^{T} \hat{Y}_j^2   -  \frac{1}{\tau^2} \frac{(T-t_0+1)}{2}.
		\end{array}
	\end{equation}

	Then for any $a_T \in (0,1)$ for some constant $C>0$ we have that 
	\begin{equation}
		\label{eqn:conclusion42}
		\mathbb{P}\left (   \vert \tilde{\epsilon}_{t_0,s}   \vert    \leq   \frac{ C \sqrt{ T }\log T}{ c_0 a_T},\,\,\,\,\,\forall s> t_0  \right)  \,\geq \, 1- a_T^2,
	\end{equation}
	
	and 
	\begin{equation}
		\label{eqn:conclusion52}
		\mathbb{P}\left (   \vert \tilde{\epsilon}_{s}   \vert    \leq   \frac{ C  \sqrt{ T }\log T }{c_0 \, a_T}  ,\,\,\,\,\,\,  \forall s>t_0 \right)  \,\geq \, 1- a_T^2,
	\end{equation}
	where $c_0$ is the smallest element of $I_2$. Furthermore, 
	\begin{equation}
		\label{eqn:conclusion62}
		\tilde{\epsilon}_{t_0}  =				\tilde{\epsilon}_{t_0,s}    + 		\tilde{\epsilon}_{s}.
	\end{equation}
	
\end{lemma}

\begin{proof}
	The claim follows with a very similar argument to that in the proof of Lemma \ref{lemma6}.
\end{proof}

\newpage

\section{Additional Simulation Studies}

\subsection{Comparison to the Bayesian method of \cite{fea06}}

\begin{table} \centering 
	\caption{\small{\textbf{Simulation study: comparison to \cite{fea06}.} Averages across $300$ repetitions for different sample sizes (T): bias $K-\widehat{K}$ (the lower, the better), Hausdorff statistics $d(\widehat{\mathcal{C}},\mathcal{C}^*)$ (the lower, the better), and time (in seconds)}. PRISCA is our proposal, the others are different runs of \cite{fea06} method (ExEff).} 
	\label{tab:bayes} 
	\scalebox{0.8}{\begin{tabular}{@{\extracolsep{5pt}} cc|ccc} 
			\\[-1.8ex]\hline 
			
			T & $a_0$ & $K-\widehat{K}$ & $d(\widehat{\mathcal{C}},\mathcal{C}^*)$  & Time  \\ 
			\hline \\[-1.8ex] 
			200 & PRISCA & 1.5 & 80.2 & 0.02  \\ 
			& ExEff-0.5 & 2.73 & 171.04 & 1.37  \\ 
			& ExEff-0.3-t & 2.24 & 135.8 & 1.18 \\ 
			& ExEff-0.3 & 2.25 & 136.1 & 1  \\ 
			\hline
			500 & PRISCA & 1.83 & 109.93 & 0.5  \\ 
			& ExEff-0.5 & 4.43 & 387.78 & 7.83  \\ 
			& ExEff-0.3-t & 3.42 & 239.87 & 7.12  \\ 
			& ExEff-0.3 & 3.49 & 249.28 & 7.74  \\ 
			\hline
			1000 & PRISCA & 2.6 & 204.18 & 4.72  \\ 
			& ExEff-0.5 & 6.28 & 769.97 & 30.02 \\ 
			& ExEff-0.3-t & 4.88 & 477.22 & 23.08  \\ 
			& ExEff-0.3 & 4.97 & 490.96 & 37.13  \\ 
	\end{tabular} }
\end{table}

We compare PRISCA to an additional Bayesian method, the one proposed by \cite{fea06} -- henceforth ExEff. The methodology outputs a posterior distribution over the change point locations (\text{e.g.}, \cite{fea06} Figure 1 (d)). Such posterior cannot be compared directly to that obtained by PRISCA because the two methods rely on different models. To extract point estimates and credible sets, one needs a post-processing steps that, to the best of our understanding, \cite{fea06} does not provide.

For this task, we employ a strategy recently used by \cite{cap21} in the context of changes in means: ``label" as change point estimates all the points that exceed a certain threshold, for example,  $0.5$ is the most natural one \citep{bar04} (ExEff-0.5). However, the choice of such a threshold is critical, as it is clear from our study here. We also include a run of the method where the threshold is equal to $0.3$ (ExEff-0.3). To construct the sets, we start from the detected change points and include the minimum number of time instances adjacent to them that makes the posterior probability of the set bigger or equal to $0.9$. We acknowledge the limitations and the degree of subjectivity in all these choices. That's why we did not include this comparison in the main manuscript.

There is a \texttt{Python} implementation of \cite{fea06} method  --available at \url{https://github.com/hildensia/bayesian_changepoint_detection}. We employ it here and compare it to our \texttt{R} implementation. ExEff also includes a truncation parameter that can be tuned to speed up the algorithm. Computational gains are obtained at the expense of moving away from exact sampling. We employ the truncation parameter suggested by Python package's authors but include a speed-up of ExEff where we modify the truncation parameter  (ExEff-0.3-t).

Table~\ref{tab:bayes} summarizes the result. The main takeaway is that PRISCA is orders of magnitude faster than ExEff (the post-processing step is not included in the runtime). This is particularly impressive given that PRISCA is implemented on \texttt{R}, while ExEff runs in \texttt{Python}, generally considered a much faster language. This is not \cite{fea06} implementation, but ExEff average runtimes align with what discussed in the paper for these sample sizes. In this specific example, PRISCA is more accurate than ExEff. Theoretically, there is no reason why PRISCA should be more accurate than ExEff. We attribute the lack of accuracy to the post-processing step. A different threshold leads to other performance (see ExEff-0.5 vis-a-vis ExEff-0.3). It is outside the scope of this study to select the best tuning parameter.




\bibliographystyle{agsm}
\begin{spacing}{1}
	\bibliography{biblio.bib}
\end{spacing}

\end{document}